\documentclass[11pt]{article}

\usepackage[letterpaper,margin=1in]{geometry}
\usepackage[T1]{fontenc}
\usepackage[utf8]{inputenc}
\usepackage{lmodern}
\usepackage{microtype}
\usepackage{setspace}
\usepackage{lineno}
\usepackage{amsmath,amssymb,amsthm}
\usepackage{booktabs}
\usepackage{graphicx}
\usepackage{array}
\usepackage{tabularx}
\usepackage{enumitem}
\usepackage{xcolor}
\usepackage[authoryear,round]{natbib}
\usepackage[hidelinks]{hyperref}
\hypersetup{
  pdftitle={AuditLock: Pre-Specified Verification Audits for Synthetic Tabular Data Releases},
  pdfauthor={Jeffery Opoku; David Banahene},
  pdfsubject={Synthetic tabular data release auditing, use-specific certificates, and release governance},
  pdfkeywords={synthetic data; data release; audit design; differential privacy; membership inference; reproducibility; statistical benchmarking},
  pdfcreator={pdfLaTeX}
}
\usepackage[nameinlink,noabbrev]{cleveref}
\crefname{subsection}{Subsection}{Subsections}
\Crefname{subsection}{Subsection}{Subsections}
\let\cite\citep
\usepackage{placeins}

\newtheorem{proposition}{Proposition}
\newtheorem{corollary}[proposition]{Corollary}
\newtheorem{remark}{Remark}
\newcommand{\E}{\mathbb{E}}

\newcommand{\sg}{\textsc{AuditLock}}
\newcommand{\eps}{\varepsilon}
\newcommand{\code}[1]{\texttt{\detokenize{#1}}}
\newcolumntype{Y}{>{\raggedright\arraybackslash}X}

\title{\textbf{AuditLock: Pre-Specified Verification Audits for Synthetic Tabular Data Releases}}
\author{
Jeffery Opoku\thanks{Corresponding author.} \qquad David Banahene\\[0.65em]
\small School of Mathematical and Statistical Sciences, The University of Texas Rio Grande Valley\\
\small \href{mailto:jeffery.opoku01@utrgv.edu}{jeffery.opoku01@utrgv.edu}\\[0.35em]
\small Department of Biostatistics, Robert Stempel College of Public Health and Social Work,\\
\small Florida International University\\
\small \href{mailto:dbana006@fiu.edu}{dbana006@fiu.edu}
}
\date{}

\begin{document}
\singlespacing
\maketitle

\begin{abstract}
A statistical office with a synthetic version of a protected file must decide whether it may be released for a stated purpose. We study the locked audit, a pre-specified and simultaneous form of the verification servers used in official statistics. Before a protected audit sample is opened, the custodian fixes the use, the candidates, the analyses, a tolerance for each and the schedule. The audit then reports simultaneous finite-sample bounds on paired loss gaps, and a declared negative control must fail. We correct an anytime-valid sequence that an earlier release had implemented with a variance penalty four times too small. Two confirmatory rounds were pre-specified in a hash-locked local ledger and publicly timestamped only after the audits. The control and ideal-generator outcomes were near-certain by design. A Gaussian copula certified in every cell, but only inside the rounding step of the tolerance rule. A second transparent candidate certified once, against a written prediction. Learned models, at library defaults and $4{,}000$ training rows, certified in one of $18$ fits in round 1 and in none of $12$ tables in round 2. The paper is a protocol and an evidence record, not a claim that any generator is private or deployable.
\end{abstract}


\section{Introduction}

Synthetic tabular data are now used for prototyping, testing, model development, data sharing, and controlled analysis. Modern generators can be sophisticated, including conditional adversarial models \cite{xu2019ctgan}, classical partially synthetic-data methods \cite{reiter2005cart}, and differentially private mechanisms \cite{dwork2006noise,mckenna2021nist}. Yet a release decision cannot be reduced to the question ``which generator gets the best average score?'' The established distinction between general and analysis-specific utility already points in this direction \cite{snoke2018utility}. A data user may need a reliable calibration curve, a subgroup-specific model assessment, or a fixed descriptive estimate. A data steward may also need evidence about a concrete attack scenario. These are different questions and they require different measurements.

This paper studies one concrete form of that evidence, a \emph{locked audit}, together with \sg{}, the open implementation and evidence record that makes it reproducible. Before a protected audit sample is opened, the custodian fixes the declared use, the candidate synthetic files, the analyses, a tolerance for each, and the audit schedule. The audit then compares each synthetic-trained workflow with its real-trained reference, reports simultaneous finite-sample bounds on the paired loss gaps, and requires a declared negative control to fail. This is a pre-specified, simultaneous version of the verification servers that statistical agencies have used to check synthetic data against confidential data (\cref{sec:related}), not a new kind of object, and it answers one question: whether a named release is supported for a named use. We do not propose a new universal synthetic-data score or a new generator. A separate release gate asks whether the audit had the required integrity conditions, and a passing metric grid is deliberately insufficient for an automatic release decision. The audit is one of four evidence layers (utility, privacy, mechanism, governance), described in Section~S30 of the supplement. The confirmatory evidence is in \cref{sec:fresh,sec:roundtwo}; the earlier studies are the history that those rounds build on.

The distinction matters. An observed low membership-inference score does not establish general privacy: attacks depend on the adversary, information, and access pattern \cite{shokri2017membership,stadler2022groundhog}. Conversely, a differential privacy claim is a formal property of a specified mechanism and adjacency relation, not a statement that every intended analysis has small error \cite{dwork2006noise}. Existing synthetic-data tools and benchmarks provide valuable generation and evaluation infrastructure, including NIST's SDNist tools and the SDGym benchmark environment \cite{nist2022sdnist,sdgym2024}. \sg{} complements rather than replaces them by requiring a locked, use-level evidence packet.

The contributions are as follows.
\begin{enumerate}[leftmargin=1.35em]
  \item \emph{Instruments.} A simultaneous finite-sample certificate for a fixed panel of bounded, paired loss gaps with a sufficient audit-size rule; variance-adaptive and anytime-valid versions, one of which corrects an error in an earlier release of the software (the anytime sequence had dropped a factor of $4$ in the variance term, so the tightening relative to the range-based bound is $2.5$ to $7.2$ times on the Adult evidence and not the $3.9$ to $9.8$ times reported with that release); and an independent-block certificate for audit units such as households. Each rests on standard inequalities, and we say so.
  \item \emph{Limits of an audit.} Unrestricted reuse of a protected audit fails; the query budget at which an audit is reconstructed depends on the decoder and on answer precision, and is about half the audit size and not the whole of it for a decoder that uses the range of the labels; and in the worst case of independent candidates the panel-size correction cannot be removed.
  \item \emph{Planning a tolerance and a control.} The ideal-generator ceiling, the smallest tolerance at which a perfect generator would be certified, computed from development data; a rule that the negative control must be \emph{demonstrably} outside tolerance and not merely fail to certify; and a re-analysis of every locked case under that rule, in which only the ACS predictive-Brier controls qualify and none of the five UCI cases does.
  \item \emph{Two confirmatory rounds, pre-specified and hash-locked} (the sense of ``pre-specified'' is defined in \cref{sec:fresh}) on audits that had not been opened, with households as the audit unit for survey data and an ideal-generator positive control. Round 1 used a tolerance rule that returned its floor of $0.02$ in every case and three outcome sentences; round 2 removed the floor, added a second workflow, and put six predictions and four outcome sentences on record (\cref{sec:fresh,sec:roundtwo}). The outcomes that discriminate are the copula's margin inside a rounding step, one failed prediction and the learned models; the control and ideal-generator outcomes were near-certain by design.
  \item \emph{An open evidence record.} An installable command-line package, hash-checked registries, a container and continuous integration; every study including failed, excluded and invalidated ones, with the disclosure that the Adult audit had been opened once before its lock.
\end{enumerate}

\Cref{fig:decisionflow} summarizes the workflow. The protected audit produces statistical and threat-model-specific evidence; it does not inherit the organization's authority to release data.

\begin{figure}[t]
\centering
\includegraphics[width=\linewidth]{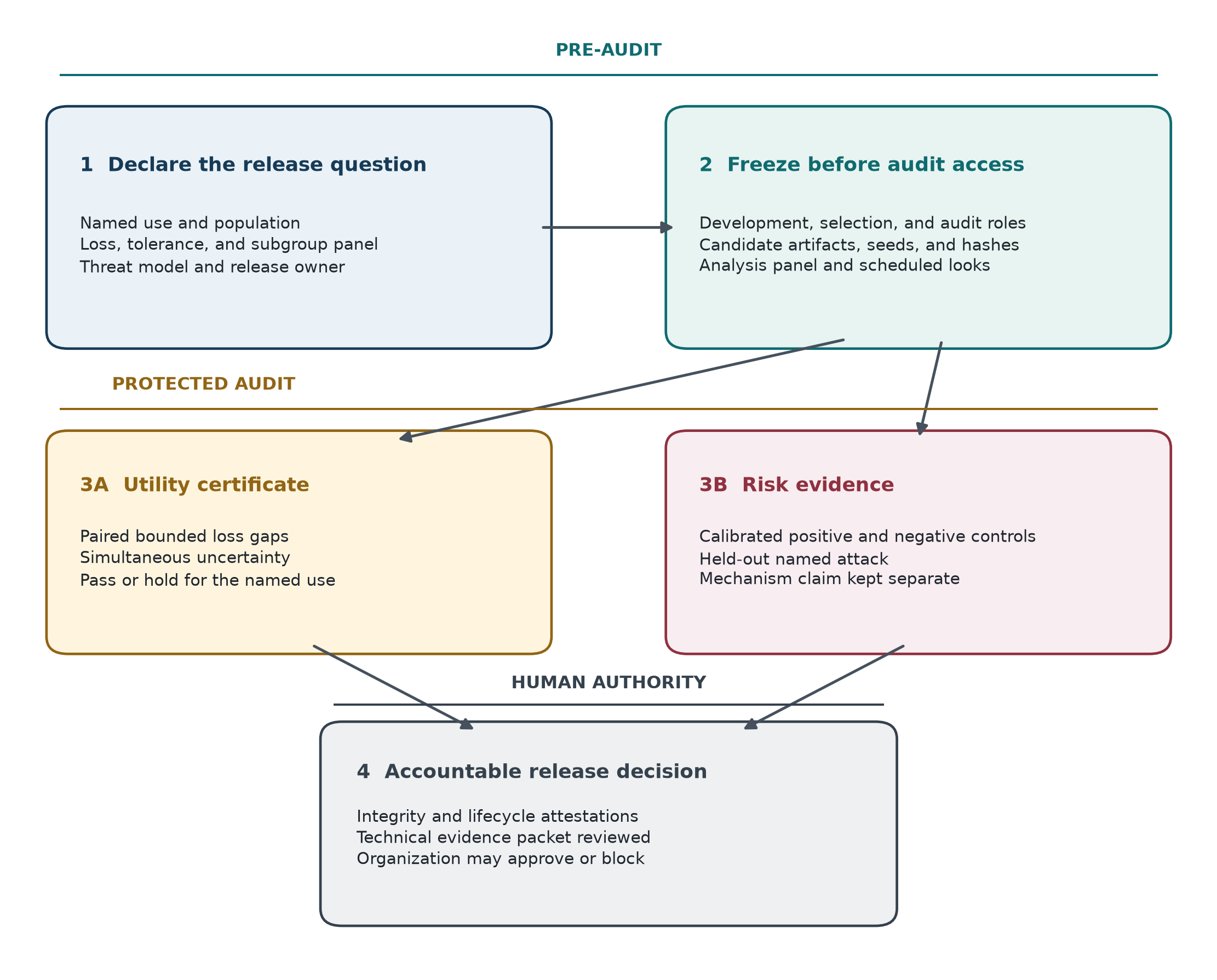}
\caption{\sg{} decision flow. The declared use, tolerance, candidate artifacts, and analyses are frozen before protected-audit access. Utility certification and empirical risk evidence remain separate, and neither can bypass the accountable release decision.}
\label{fig:decisionflow}
\end{figure}

\paragraph{Scope.} The public studies in this paper use ACS Public Use Microdata Sample files. The Census Bureau documents their variable coding, sample design, and weights \cite{census2024pums}. The earlier public studies use a row-level independent-and-identically-distributed (iid) working assumption to demonstrate the technical protocol. They are not design-based survey inference, are not evidence about household clustering, and are not a substitute for a live confidential audit. The confirmatory study of \cref{sec:fresh} uses households as the audit unit and reports a replicate-weight sensitivity check, but it is still a unit-level technical audit and not design-based inference for the survey population. The paper makes no universal claim that any method is safe, private, or ready for deployment.

\section{Related Work and Positioning}
\label{sec:related}

Partially synthetic data have a long history in disclosure control and official statistics \cite{reiter2005cart}, and practical sequential-regression synthesis is available in widely used tools such as \textsf{synthpop} \cite{nowok2016synthpop}. Recent tabular generators include deep conditional models such as CTGAN \cite{xu2019ctgan}. Differential privacy bounds how much a mechanism's output distribution can change under a stated neighboring-dataset relation \cite{dwork2006noise}, and private Bayesian-network and private-GAN approaches show the range of possible generators \cite{zhang2017privbayes,jordon2019pategan}. None of this makes the evaluation problem disappear, because the usefulness of a release remains use-dependent.

\paragraph{Verification and validation servers.} The closest statistical antecedent of a locked audit is the verification server. The custodian keeps the confidential file and releases a synthetic version, and an analyst's named analysis is run on both, with a measure of how far the synthetic-data inference lies from the confidential-data inference returned to the analyst \cite{reiter2009verification,mcclure2012feedback}. Differentially private versions were developed for regression predictions \cite{yu2018verification} and for survey-weighted estimates \cite{lin2024survey}; a synthesis-and-verification system was built for data on employees of the U.S.\ federal government \cite{barrientos2018synthesis}; and a feasibility study evaluates differentially private \emph{validation} servers, in which user code is run on the confidential data after it was developed on public data, using administrative and survey data \cite{barrientos2024feasibility}. The U.S.\ Census Bureau's SIPP Synthetic Beta follows the validation pattern: users first run their programs on a synthetic data server, then ask Census staff to run them on the confidential data, with the output subject to disclosure review \cite{census2017ssbguidelines}. What we add to this tradition is narrower than a new kind of object. The whole evaluation is fixed before the protected data are touched (candidate files, analyses, tolerances, audit schedule); the answer is a simultaneous finite-sample bound over the whole panel and not a per-analysis summary; and a negative control must be shown to fail. We claim a pre-specified, simultaneous version of use-specific verification, and an analysis of when its decisions are informative.

\paragraph{Statistical building blocks and neighbouring work.} A finite-sample guarantee over a finite set of candidates with bounded losses is the setting of Learn then Test \cite{angelopoulos2025learn}; our certificate differs in testing a paired gap against a tolerance across audit looks, but it rests on the same kind of argument. The estimand compares the train-on-synthetic, test-on-real idea \cite{esteban2017rcgan} with the same workflow trained on real data. Recent work asks when synthetic data are suitable for a given use \cite{kulynych2026suitability}, proposes analysis-specific evaluation of differentially private synthetic data \cite{arnold2020really}, assesses and improves utility measures \cite{raab2021utility}, debiases inference from synthetic data \cite{decruyenaere2024debiasing}, and evaluates models after selection with e-values \cite{farzaneh2026postselection}; Drechsler and Haensch review utility and risk measurement for official statistics \cite{drechsler2024thirty}. We have not reproduced these methods and make no claim of superiority over them. The word ``certificate'' here means a finite-sample equivalence statement for a declared loss and tolerance, and is unrelated to certified private data release \cite{donhauser2024certified}.

Benchmarking libraries are essential for comparing generators. SDGym supplies a general benchmarking environment and SDNist supplies data and metrics built for structured tabular evaluation \cite{sdgym2024,nist2022sdnist}. SDMetrics reports validity, structure, marginal, pairwise, and task-oriented measurements \cite{sdmetrics2026}; SynthEval combines configurable fidelity, utility, and privacy metrics \cite{lautrup2024syntheval}; Anonymeter provides attack-based singling-out, linkability, and inference-risk evaluations \cite{giomi2022anonymeter}. Systematic assessments also exist: Du and Li evaluate eight synthesizer families on twelve real data sets and propose improved metrics \cite{du2025systematic}. NIST's public challenge experience illustrates the value of explicit evaluation design when comparing privacy-preserving synthesis \cite{ridgeway2021challenge}, and earlier statistical work framed utility as a comparison of inferences drawn from protected and original data, often alongside disclosure risk \cite{karr2006utility}. On the privacy side, membership inference asks whether a record was in the training data \cite{shokri2017membership}, and broad empirical studies show that synthetic-data privacy and utility are hard to predict across models and data sets \cite{stadler2022groundhog}. \sg{} is designed around that unpredictability rather than against it, and it does not turn one attack result into a global privacy score.

The closest prior work in spirit is Houssiau et al.'s framework for auditable synthetic data generation \cite{houssiau2022auditable}. There a generator is restricted to a small set of well-chosen statistics and the release carries a \emph{generator card} declaring which information was used, and a companion procedure, reduced to a hypothesis test in a regression problem, lets a data holder verify that the card is accurate. The same group's TAPAS toolbox supplies adversarial privacy auditing \cite{houssiau2022tapas}. Both share our conviction that a release should arrive with checkable evidence rather than assurances, and the two audits are complementary because they answer different questions. Houssiau et al. audit the \emph{generator}: given a declared card, did synthesis consume only the approved inputs? \sg{} audits the \emph{proposed release for a declared downstream use}: on a protected evaluation locked before it was opened, does a synthetic-trained workflow match its real-trained reference within a pre-committed tolerance, simultaneously across a whole candidate-by-analysis panel, with a declared negative control required to fail and a governance gate able to withhold a technically passing release? Neither implies the other. A generator can consume only approved statistics and still be useless for the named analysis, and a release can pass our equivalence audit while its card is silent about what it consumed. An organization plausibly wants both.

The difference from a systematic assessment is procedural rather than a matter of which metric is used. We require the candidate panel, analysis panel, tolerances, and audit schedule to be committed and hashed before the protected audit is opened; we report a simultaneous finite-sample bound over the whole candidate-by-analysis panel rather than per-metric point estimates; and we require a declared negative control to fail before any result is admitted. The accompanying five-scenario audit makes the workflow checks executable: audit reuse recorded after the lock, a passing negative control, unregistered post-hoc metrics, and a missing attestation each trigger a \sg{} hold. It tests workflow coverage of the software and is not a comparison of numerical accuracy with SDMetrics, SynthEval or Anonymeter, which remain the more specialized fidelity and privacy-risk instruments. The object is therefore not another realism score or generator leaderboard. It is a decision record that couples a predeclared use, a protected audit, simultaneous uncertainty, calibrated controls, threat-model-specific risk evidence, and accountable release authority, without letting any one component stand in for the others.

\section{Problem Setup}

Let $D_{\mathrm{dev}}$ be a development sample, $D_{\mathrm{sel}}$ a generator-selection sample, and $D_{\mathrm{audit}}=\{Z_i\}_{i=1}^n$ an audit sample. A candidate generator $g\in\mathcal{G}$ is fit using only the data allowed by its declared protocol. A synthetic data set $S_g$ is produced before audit evaluation. For a declared analysis $a\in\mathcal{A}$, let $f_{R,a}$ be the workflow trained or calibrated on its permitted real-data reference and let $f_{S,g,a}$ be the corresponding workflow trained or calibrated on $S_g$. The exact construction of each workflow is part of the locked record.

The primary estimand is a use-level expected loss gap,
\begin{equation}
 \Delta_{g,a}=\E\left[\ell_a(f_{S,g,a},Z)-\ell_a(f_{R,a},Z)\right],
 \label{eq:gap}
\end{equation}
where the loss is transformed in advance to a declared bounded scale. The audit estimator is
\begin{equation}
 \widehat{\Delta}_{g,a}=\frac{1}{n_a}\sum_{i\in I_a}\left\{\ell_a(f_{S,g,a},Z_i)-\ell_a(f_{R,a},Z_i)\right\}.
 \label{eq:gapestimate}
\end{equation}
The set $I_a$ may be a locked subgroup. Thus a subgroup is not a post-hoc plot: it is a separate analysis with its own sample size, tolerance, and uncertainty bound.

The statistical target is equivalence within a declared tolerance, not superiority of a generator and not rejection of a point null. The logic is confidence-set inclusion, as in classical equivalence testing \cite{schuirmann1987}. A candidate is technically certified for analysis $a$ only when the simultaneous upper bound satisfies $|\widehat{\Delta}_{g,a}|+h_a\leq\tau_a$. Otherwise the outcome is a \emph{hold}: it may reflect a material gap, insufficient audit information, or both. Unless a separate lower bound excludes $\tau_a$, a hold is not evidence that $|\Delta_{g,a}|>\tau_a$. Throughout the results, ``fails a check'' means ``fails to certify under this rule.''

The public protocol locks the data source, splits, generator versions, generation seeds, candidate grid, analysis panel, tolerances, and scheduled audit looks before headline evaluation. The sealed version additionally keeps audit records unavailable to generator developers. This separation is essential. Once an audit is repeatedly inspected to choose models or analyses, the simple finite-sample statement below no longer describes the full adaptive process; holdout reuse and adaptivity require their own validity analysis \cite{dwork2015holdout}.

\section{A Simultaneous Use-Specific Certificate}

This section states the primary certificate, then delimits it. \Cref{sec:primarycertificate} gives the certificate and the planning rules that follow from it. \Cref{sec:lowerbounds} establishes what no version of it can support: an audit opened to unrestricted reuse is reconstructible, and the panel-size correction cannot be dropped. \Cref{sec:varianceadaptive} sharpens the reported width without weakening the guarantee, and \cref{sec:auditunits} relaxes the row-level independence assumption for clustered and time-ordered records.

\subsection{The primary certificate and its planning rules}
\label{sec:primarycertificate}

We state the primary certificate in a simple form so that its guarantees and its limits are visible. It uses standard concentration, not a new probability inequality.

\begin{proposition}[Scheduled simultaneous bounded-loss audit]
\label{prop:certificate}
Fix a finite candidate panel $\mathcal{G}$, analysis panel $\mathcal{A}$, and a pre-specified sequence of audit looks $l\in\{0,\ldots,L-1\}$. Conditional on all information fixed before the audit, suppose that for every candidate, analysis, and scheduled look, the paired losses
\[
 X_{i,g,a}=\ell_a(f_{S,g,a},Z_i)-\ell_a(f_{R,a},Z_i)
\]
are iid and lie in $[-B_a,B_a]$. Let $q_l>0$ be fixed spending weights with $\sum_l q_l\leq 1$, and let $n_{a,l}$ be the scheduled audit size for analysis $a$ at look $l$. Define
\begin{equation}
 h_{a,l}=B_a\sqrt{\frac{2\log\{2|\mathcal{G}||\mathcal{A}|/(\alpha q_l)\}}{n_{a,l}}}.
 \label{eq:hoeffding}
\end{equation}
Then, with probability at least $1-\alpha$, all declared candidates, analyses, and scheduled looks satisfy
\[
 \left|\widehat{\Delta}_{g,a,l}-\Delta_{g,a}\right|\leq h_{a,l}.
\]
Consequently, the rule $|\widehat{\Delta}_{g,a,l}|+h_{a,l}\leq\tau_a$ implies $|\Delta_{g,a}|\leq\tau_a$ on that event.
\end{proposition}

\begin{proof}
For a fixed $(g,a,l)$, Hoeffding's inequality for variables in $[-B_a,B_a]$ gives
\[
 \Pr\left(\left|\widehat{\Delta}_{g,a,l}-\Delta_{g,a}\right|>h_{a,l}\right)
 \leq \frac{\alpha q_l}{|\mathcal{G}||\mathcal{A}|}.
\]
A union bound over candidates, analyses, and scheduled looks gives a total error probability at most $\alpha\sum_lq_l\leq\alpha$. The final statement follows from the triangle inequality.
\end{proof}

No independence is required across candidates or analyses: the same audit records may induce strong dependence among their paired losses. The union bound only requires the stated iid condition across audit units for each fixed bounded loss sequence.

\begin{corollary}[Selection before audit]
\label{cor:selection}
Under the conditions of \cref{prop:certificate}, let $\widehat g$ be selected by any rule measurable with respect to development data, selection data, pre-audit randomness, and the frozen candidate artifacts, but not the protected audit outcomes. Then the certificate in \cref{prop:certificate} holds for every declared analysis of $\widehat g$ with probability at least $1-\alpha$.
\end{corollary}

\begin{proof}
Conditional on all pre-audit information, $\widehat g$ is one member of the fixed panel $\mathcal G$. The simultaneous event in \cref{prop:certificate} holds for every member of that panel, so it holds for the selected member.
\end{proof}

\begin{proposition}[Sufficient audit size for a certification margin]
\label{prop:power}
Under the conditions of \cref{prop:certificate}, fix a candidate, analysis, and scheduled look. Suppose the true gap lies inside its tolerance by a known planning margin $0<\eta_a\leq\tau_a$,
\[
 |\Delta_{g,a}|\leq \tau_a-\eta_a.
\]
If
\begin{equation}
 n_{a,l}\geq
 \frac{8B_a^2}{\eta_a^2}
 \log\!\left\{\frac{2|\mathcal G||\mathcal A|}{\alpha q_l}\right\},
 \label{eq:powerplanning}
\end{equation}
then the candidate passes the declared check on the joint $1-\alpha$ event from \cref{prop:certificate}.
\end{proposition}

\begin{proof}
Condition \eqref{eq:powerplanning} makes $h_{a,l}\leq\eta_a/2$. On the joint concentration event,
\[
 |\widehat\Delta_{g,a,l}|+h_{a,l}
 \leq |\Delta_{g,a}|+2h_{a,l}
 \leq \tau_a-\eta_a+\eta_a=\tau_a.
\]
\end{proof}

The bound is a pre-audit planning result, not observed power and not a promise that an unknown candidate will pass. It makes the cost of a larger candidate panel, a larger analysis panel, or more scheduled looks explicit. The command-line tool implements the equal-spending version of \eqref{eq:powerplanning}.

\begin{proposition}[Valid correction for controlled audit adaptation]
\label{prop:controlledadaptation}
Suppose that, before audit access, the protocol enumerates a finite menu $\mathcal T$ of at most $M$ permitted adaptation transcripts. Each transcript fixes a bounded gap function for every declared candidate, analysis, and look, but the final transcript may be selected after inspecting the audit. Replacing \eqref{eq:hoeffding} by
\[
 h^{(M)}_{a,l}=B_a\sqrt{\frac{2\log\{2M|\mathcal G||\mathcal A|/(\alpha q_l)\}}{n_{a,l}}}
\]
gives simultaneous coverage for every transcript and therefore for the audit-selected transcript.
\end{proposition}

\begin{proof}
Apply Hoeffding's inequality to every transcript, candidate, analysis, and scheduled look, allocating $\alpha q_l/(M|\mathcal G||\mathcal A|)$ to each event. A union bound gives total error probability at most $\alpha\sum_lq_l\leq\alpha$. Because the event holds for all $t\in\mathcal T$, it also holds for the transcript chosen after audit inspection.
\end{proof}

\Cref{prop:controlledadaptation} provides a narrow correction, not permission for open-ended tuning. The adaptation branches must be concrete enough to enumerate and reproduce before audit access. \Cref{fig:theoryplanning} shows the resulting planning cost for 35 base checks.

\begin{figure}[t]
\centering
\includegraphics[width=0.94\linewidth]{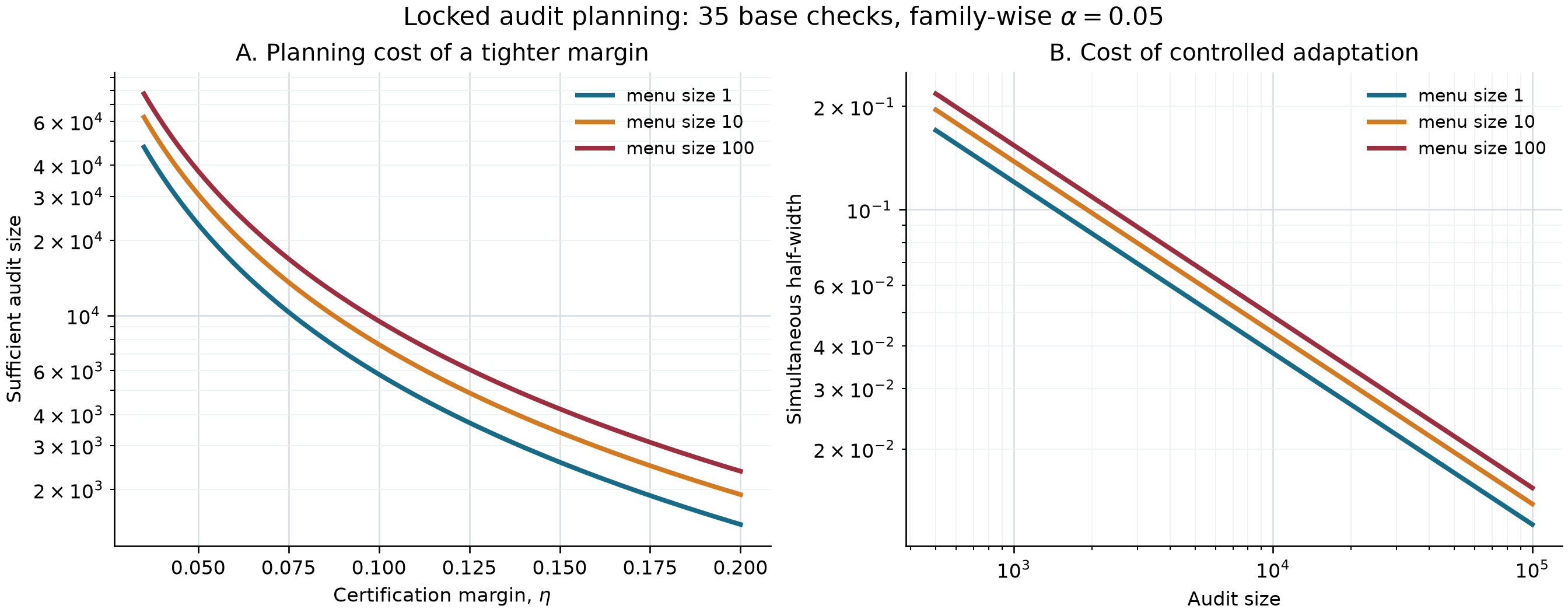}
\caption{Pre-audit planning under bounded paired losses. Panel A gives the sufficient audit size in \cref{prop:power} for 35 base checks as the certification margin changes. Panel B shows the simultaneous half-width when the fixed adaptation menu in \cref{prop:controlledadaptation} has size 1, 10, or 100. The logarithmic adaptation cost is much smaller than the quadratic cost of halving the margin, but it is not zero.}
\label{fig:theoryplanning}
\end{figure}

For the completed public ACS studies, losses were scaled so $B_a=1$. We used the pre-specified spending schedule $q_l=6/\{\pi^2(l+1)^2\}$, which sums to one over infinitely many planned looks. The public runs report the whole candidate grid. Selecting a candidate on $D_{\mathrm{sel}}$ is permitted by the protocol; selecting it after observing $D_{\mathrm{audit}}$ is not.

\begin{remark}[What the certificate means]
\label{rem:meaning}
The certificate is a statement about the declared loss, tolerance, population represented by the iid audit assumption, and frozen workflow. It is not a confidence statement about arbitrary downstream analyses. It does not establish that the synthetic rows are private, non-identifying, legally releasable, survey-design valid, or appropriate for an untested subgroup. A failure can arise from real utility loss, an overly strict tolerance, a small audit sample, or a conservative bound. In all of these cases, the correct outcome is \emph{hold or uncertified}, not a forced pass/fail narrative about a generator.
\end{remark}

\subsection{Limits of an audit}
\label{sec:lowerbounds}

\begin{proposition}[Failure under unrestricted audit reuse]
\label{prop:reuseimpossibility}
Let $P$ be a non-atomic audit distribution on a standard Borel space and let $S=(Z_1,\ldots,Z_n)$ be the observed audit sample. If an analyst may choose an arbitrary bounded gap function after seeing $S$, there is a measurable choice $X_S:\mathcal Z\rightarrow[0,1]$ with empirical audit mean zero and population mean one almost surely. Hence, for any $h\leq\tau<1$, a rule that accepts when $|\widehat\Delta|+h\leq\tau$ can accept while the true gap exceeds tolerance.
\end{proposition}

\begin{proof}
After observing $S$, define $X_S(z)=0$ when $z\in\{Z_1,\ldots,Z_n\}$ and $X_S(z)=1$ otherwise. Its empirical mean on $S$ is zero. Because $P$ is non-atomic, the finite observed set has probability zero, so $\E_P[X_S(Z)]=1$ almost surely. Therefore $|\widehat\Delta|+h=h\leq\tau$, while $|\Delta|=1>\tau$.
\end{proof}

This is elementary, the standard holdout-reuse observation. It does not say that all adaptive analysis is impossible, only that unrestricted, unrecorded reuse cannot inherit the locked-audit guarantee. A finite, pre-declared adaptation menu can be handled by including every permitted branch in the simultaneous panel; richer reuse needs a reusable-holdout or adaptive-inference design \cite{dwork2015holdout}.

The obvious objection is that the witness function is pathological, since no real analyst writes down the indicator of the observed audit set. The next result removes that objection: the same failure arises from ordinary bounded queries under a realistic budget, so the danger is operational rather than measure-theoretic.

\begin{proposition}[Reconstruction from bounded audit queries]
\label{prop:reconstruction}
Assume the analyst knows which records constitute the protected audit, that is their identifiers, but not their hidden binary labels $b\in\{0,1\}^{n}$, and may submit bounded gap functions and read back only their reported empirical audit means to full precision. For $j=1,\ldots,m$, let $s_j\in\{-1,1\}^{n}$ be a publicly seeded pseudorandom sign vector indexed by identifier, and let query $j$ report the bounded statistic
\[
 y_j=\frac1n\sum_{i=1}^{n}s_{ji}\,(2b_i-1)\in[-1,1].
\]
Write $S\in\{-1,1\}^{m\times n}$ for the stacked signs. If $m\geq n$, then $S$ has full column rank with probability at least $1-(1/2+o(1))^{n}$, and on that event $b$ is determined exactly by $y=S(2b-\mathbf 1)/n$.
\end{proposition}

\begin{proof}
The map $b\mapsto S(2b-\mathbf 1)/n$ is affine and injective whenever $S$ has full column rank, in which case inverting recovers $2b-\mathbf 1$ and hence $b$. For $m\geq n$ it suffices that the leading $n\times n$ block be non-singular, and the probability that a square Rademacher matrix is singular is at most $(1/2+o(1))^{n}$ \cite{tikhomirov2020singularity}.
\end{proof}

The identifier assumption is what makes the attack ordinary, and it is the standard setting of database reconstruction \cite{dinur2003revealing}: audit membership is often administratively known, for example a published list of sampled record keys, while the sensitive labels are exactly what the audit protects. Full column rank is a high-probability event rather than a certain one, since a Rademacher design is singular with positive probability, so the statement cannot be strengthened to hold almost surely.

The consequence for a release audit is a certificate that accepts a memorising candidate. Let the declared analysis be zero-one prediction loss, let $f_R$ be the real-trained reference with audit loss $\widehat L(f_R)$, and let $f_S$ return the recovered label on any audit identifier and a fixed class elsewhere. On the reconstruction event $f_S$ has audit loss zero, so its empirical gap is $\widehat\Delta=-\widehat L(f_R)$ and the rule $|\widehat\Delta|+h\leq\tau$ accepts whenever $\widehat L(f_R)\leq\tau-h$, the ordinary situation of a competent reference. Its population gap is $\Delta=L(\mathrm{const})-L(f_R)$, the excess loss of the constant rule, which is strictly positive whenever the reference beats that rule and exceeds $\tau$ for any reference worth auditing. The audit therefore certifies a lookup table.

\Cref{prop:reconstruction} is the release-audit specialization of classical database reconstruction \cite{dinur2003revealing} and of adaptive data analysis lower bounds \cite{hardt2014preventing,steinke2015interactive}. Linear reconstruction against synthetic data has also been used as an attack in its own right \cite{annamalai2024linear}. \Cref{fig:reconstruction} reports the empirical threshold for the least-squares decoder. With $n=400$ audit rows, mean recovery is $0.64$ at a budget of $0.1n$ queries and $0.84$ at $0.5n$, with no exact recoveries; exact recovery occurs in half of replicates at $0.9n$ and in every replicate once the budget reaches $n$.

That threshold near $n$ is a property of the decoder and not of the audit. Proposition~\ref{prop:reconstruction} is only a sufficient condition, and an analyst who also uses the known range of the labels recovers them with fewer queries. In a fixed-seed sweep (Section~S25 of the supplement; $n=200$, $400$ and $800$, $20$ replicates per cell) a linear program that imposes the box constraint $-1\leq x\leq1$ recovered all labels exactly in every replicate at $0.6n$ queries and in at least $90\%$ of replicates at $0.55n$, consistent with the threshold of about one half reported for linear-programming recovery of binary solutions of linear systems \cite{mangasarian2011probability}; the least-squares decoder needs $n$. The sweep also shows what answer precision does. Rounding the reported means to a step of $0.1/\sqrt n$ (checked at $n=200$ and $800$; at $n=400$ that rounding changes nothing) does not stop an $\ell_1$ decoder, which still recovers the labels exactly at $0.55n$ to $0.6n$ queries, and rounding to $0.3/\sqrt n$ delays it only to $0.7n$; rounding to $1/\sqrt n$, or Gaussian noise of standard deviation $0.3/\sqrt n$ or more, left exact recovery in at most one of twenty replicates at every budget up to $1.2n$. So the number of queries is not the quantity that protects an audit: what matters is how much information the answers carry, that is, their precision times their number. Two operational consequences follow. Query budgets against a protected audit must be counted and capped rather than left implicit, and an audit that has answered on the order of $n/2$ adaptive queries at full precision should be treated as spent, exactly as a locked panel is spent once inspected. Coarsening or adding noise to the answers moves the threshold, and doing so in a controlled way is what differentially private verification servers do \cite{yu2018verification,lin2024survey}.

\begin{figure}[t]
\centering
\includegraphics[width=0.86\linewidth]{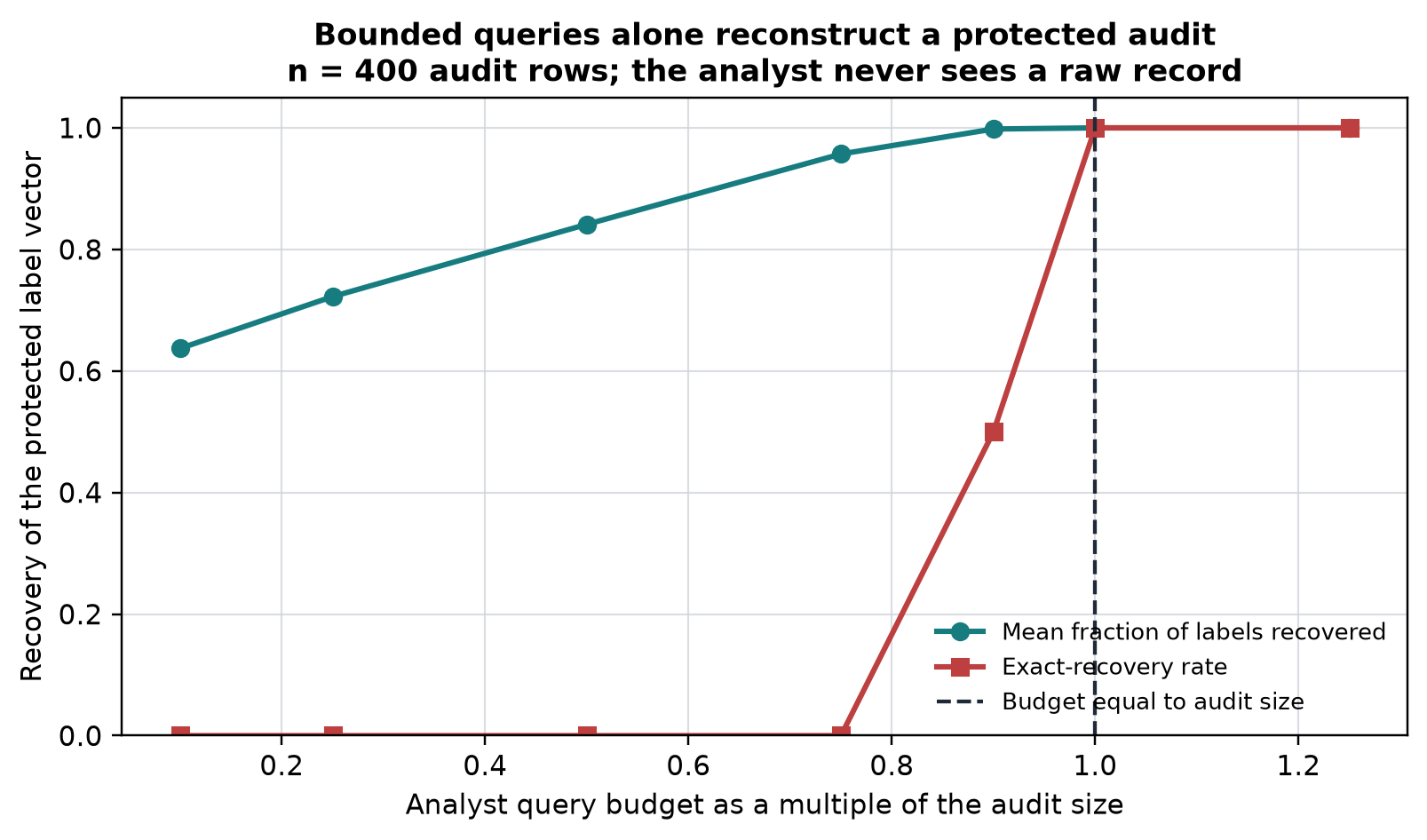}
\caption{Operational audit-reuse attack. The analyst sees only reported bounded query means, never a raw audit row. Recovery of the protected label vector is partial while the query budget is below the audit size and becomes exact at a budget comparable to it, matching \cref{prop:reconstruction}. The dashed line marks a budget equal to the audit size.}
\label{fig:reconstruction}
\end{figure}
\FloatBarrier

A companion question is whether the multiplicity correction in \cref{prop:certificate} is an artifact of the union bound. In the worst case it is not.

\begin{proposition}[Worst-case necessity of the panel correction]
\label{prop:multiplicitynecessary}
Fix $n\geq1$ and $\alpha\in(0,1/2)$. Let the panel consist of $K$ candidates whose paired gaps are iid Rademacher with true gap zero, independent across candidates, and let $\bar X_k$ be the $k$th sample mean. Consider any rule that reports the interval $\bar X_k\pm h$ for every candidate, with a common half-width $h$ that does not depend on the data. If the rule attains simultaneous coverage $1-\alpha$, then
\[
 h\ \geq\ c_\alpha\,\min\Bigl\{1,\ \sqrt{\tfrac{\log K}{n}}\Bigr\},
\]
where $c_\alpha>0$ depends only on $\alpha$. In particular a half-width that ignores $K$ cannot retain simultaneous coverage as $K$ grows.
\end{proposition}

The proof combines the independence of the $K$ candidate means with a binomial lower tail bound; it is given in Section~S15 of the supplement. The statement is a worst case with two restrictions that matter: the candidates are independent, and the intervals are centred at the sample means with a data-free width. Neither describes every audit. Real candidates share audit rows and are dependent, and in the California frontier the three DP histograms and the non-private comparator have almost identical gaps and behave much like one candidate, for which a correction for five is conservative. The truncation at one matters, because the gaps lie in $[-1,1]$, so $h=1$ always covers and no lower bound may exceed it. The result says that when candidates are close to independent and $\log K\lesssim n$, the logarithmic panel cost cannot be avoided by a sharper analysis of the same intervals.

Figure~S13 in the supplement confirms this numerically at $n=1{,}000$: a half-width computed for a single candidate retains $0.995$ simultaneous coverage at $K=1$, but falls to $0.926$ at $K=10$, $0.494$ at $K=100$, and $0.001$ at $K=1{,}000$, while the panel-corrected width from \cref{prop:certificate} stays at or above $0.995$ throughout. The logarithmic cost of a large declared panel is real in this worst case, and it is the price of being allowed to lock many candidates at once.

\subsection{Variance-adaptive instruments for the same audit}
\label{sec:varianceadaptive}

\Cref{prop:certificate} pays for the full declared range $[-B_a,B_a]$ whatever the data look like. Paired loss gaps in practice are concentrated near zero, because the real-trained and synthetic-trained workflows agree on most records and disagree on a minority. The range is then a poor proxy for the uncertainty, and the certificate is wider than the evidence requires. Two standard instruments fix this without weakening the guarantee.

\begin{proposition}[Variance-adaptive scheduled certificate]
\label{prop:eb}
Under the conditions of \cref{prop:certificate}, let $\widehat V_{g,a,l}$ be the sample variance of the paired gaps in the cell $(g,a,l)$ and set $\delta_{a,l}=\alpha q_l/(|\mathcal G||\mathcal A|)$. Replacing \eqref{eq:hoeffding} by the empirical Bernstein half-width
\begin{equation}
 h^{\mathrm{EB}}_{g,a,l}=\sqrt{\frac{2\widehat V_{g,a,l}\log(4/\delta_{a,l})}{n_{a,l}}}+\frac{14B_a\log(4/\delta_{a,l})}{3(n_{a,l}-1)}
 \label{eq:eb}
\end{equation}
preserves the conclusion of \cref{prop:certificate} with the same $1-\alpha$ simultaneous guarantee.
\end{proposition}

\begin{proof}
Maurer and Pontil's inequality is stated for observations in an interval of unit length \cite{maurer2009empirical}. Rescale the gaps by their range $R_a=2B_a$, so that $X/R_a$ lies in an interval of unit length with sample variance $\widehat V/R_a^2$. Applying the one-sided inequality to each tail at level $\delta_{a,l}/2$ and multiplying through by $R_a$ leaves the variance term unchanged, because $R_a\sqrt{\widehat V/R_a^{2}}=\sqrt{\widehat V}$, and multiplies the second term by $R_a=2B_a$, which is the source of the factor $14=7\times 2$ in \eqref{eq:eb}. Union bounding over candidates, analyses, and looks exactly as in \cref{prop:certificate}, the spending weights again give total error at most $\alpha\sum_lq_l\leq\alpha$.
\end{proof}

The range factor matters in practice as well as in principle: with $B_a=1$ the second term is $14\log(4/\delta)/\{3(n-1)\}$, roughly $0.0025$ at the final Adult look, so dropping the factor of two would understate the reported half-width by about a fifth. The implementation and every number below use the range $R_a=2$.

A second instrument removes the look schedule altogether. A predictable plug-in empirical Bernstein confidence sequence \cite{waudbysmith2024betting} is valid simultaneously at every sample size, so it may be evaluated at any stopping time, including a data-dependent decision to stop the audit. Multiplicity across the panel is still paid by dividing $\alpha$, but no alpha spending across looks is needed, because time-uniformity is built into the construction.

\paragraph{A correction to an earlier release.} Versions of this software up to 0.1.14, and earlier public descriptions of it, computed the sequence with a variance term one quarter of the size that Theorem~2 of \citet{waudbysmith2024betting} requires: that theorem uses $v_t=4(U_t-\widehat\mu_t)^2$ together with $\psi_E(\lambda)=\{-\log(1-\lambda)-\lambda\}/4$, and the earlier code kept the division by four but dropped the four. The construction in Section~S12 and the code now follow the theorem. Section~S19 of the supplement reports the time-uniform miscoverage of both forms in four fixed-seed scenarios. The earlier form exceeded its nominal level by more than three binomial standard errors on two-point data at the range endpoints, a fair coin (it missed in $18.6\%$ of runs at nominal level $5\%$) and a skewed spike ($7.6\%$ at $5\%$); it was conservative on data resembling the Adult gaps, where it missed in $1.4\%$ of runs. The corrected form stayed below nominal level in every scenario. On the Adult panel the correction widens the sequences by a factor of $1.37$ to $1.57$ at the final look, and every number below is the corrected one.

Both instruments must be pre-specified like any other element of the locked panel. Neither may be applied after the fact to revise a recorded outcome: choosing the bound family after seeing which one passes is precisely the adaptation that \cref{prop:reuseimpossibility} forbids. We therefore report them as a measurement on the already-locked Adult evidence, with every recorded pass and fail left untouched, and we recommend them for pre-specification in future panels. Because the comparison applies three instruments to identical evidence, it does not depend on the Adult audit's status, which \cref{sec:adultpanel} qualifies.

\Cref{fig:boundfamily} reports that measurement. The gap vectors are recomputed from the byte-reproducible locked artifacts, and the recomputed Hoeffding half-widths reproduce the locked certificate to within $10^{-9}$, which is the gate that makes this a comparison of instruments on identical evidence, not a new experiment. At the final $16{,}000$-record look the Hoeffding half-width is $0.0316$ for every candidate. The scheduled empirical Bernstein half-widths range from $0.0067$ for the Gaussian copula to $0.0138$ for TVAE, a tightening of $2.3$ to $4.7$ times; the corrected anytime confidence sequence ranges from $0.0044$ to $0.0129$, a tightening of $2.5$ to $7.2$ times. The ordering is informative: the copula tightens most because its gaps are the most concentrated, and TVAE least because its gaps are the most dispersed, so the variance-adaptive width reports something about each candidate that the range-based width cannot.

The practical consequence is a much larger effective audit. Under \cref{prop:power} the audit size needed for a target margin scales as the squared half-width, so a $k$-fold tightening corresponds to roughly a $k^2$-fold reduction in required audit rows at fixed margin, or a proportionally finer certifiable tolerance at fixed sample size. For the locked Adult panel the Hoeffding width consumed slightly more than half of the $0.060$ tolerance before any evidence was considered; for the copula the empirical Bernstein width consumes about a ninth of it.

\begin{figure}[t]
\centering
\includegraphics[width=0.94\linewidth]{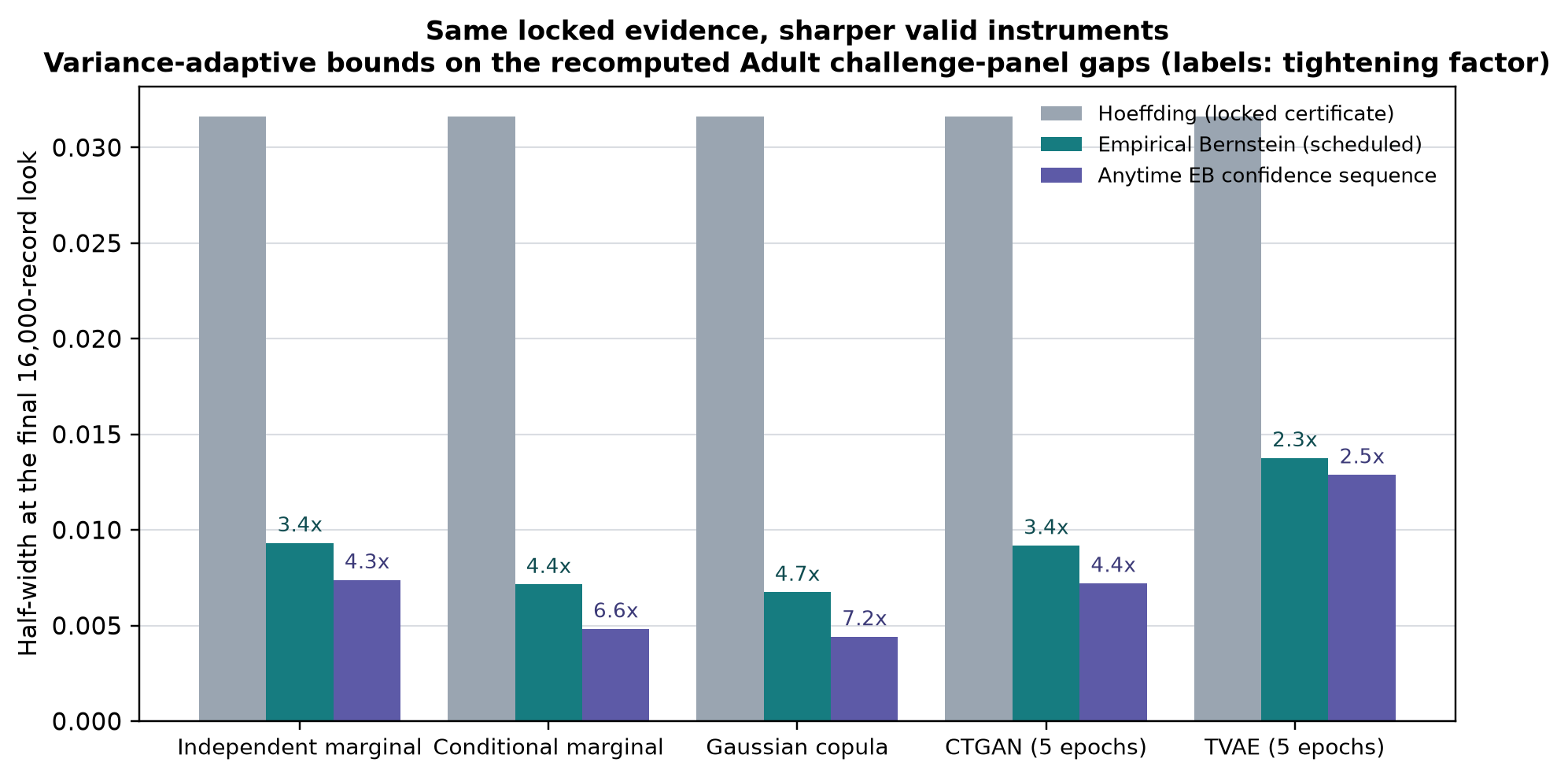}
\caption{Same locked evidence, sharper valid instruments. Half-widths at the final $16{,}000$-record look of the Adult challenge panel, recomputed from the locked artifacts. Labels give the tightening factor relative to the Hoeffding certificate. The comparison revises no recorded outcome; it quantifies what a pre-specified variance-adaptive instrument would have delivered.}
\label{fig:boundfamily}
\end{figure}
\FloatBarrier

\subsection{Audit units beyond the row}
\label{sec:auditunits}

The certificate so far assumes the audit rows are exchangeable. Two common designs break that assumption in different ways, and each needs its own audit unit rather than a correction factor.

\begin{proposition}[Scheduled simultaneous independent-block audit]
\label{prop:blockcertificate}
Fix the same candidate, analysis, and scheduled-look panels. For each analysis $a$, suppose that its pre-specified audit blocks $b=1,\ldots,m_{a,l}$ are independent and identically distributed conditional on the locked development process, and that the within-block mean paired loss gap $\bar X_{b,g,a}$ lies in $[-B_a,B_a]$. Replacing $n_{a,l}$ in \cref{eq:hoeffding} by $m_{a,l}$ gives a simultaneous bound for the equally block-weighted estimand $\Delta^{\mathrm{block}}_{g,a}=\E[\bar X_{b,g,a}]$. Thus the rule $|\widehat\Delta^{\mathrm{block}}_{g,a,l}|+h^{\mathrm{block}}_{a,l}\leq\tau_a$ is valid on the same $1-\alpha$ event.
\end{proposition}

\begin{proof}
For each fixed candidate, analysis, and look, the independent bounded block means satisfy Hoeffding's inequality. Applying the same alpha-spending allocation and union bound as in \cref{prop:certificate} proves the claim.
\end{proof}

\begin{remark}[Block scope]
The block certificate changes the estimand: every independent block has equal weight, not every row. It requires the same named blocks for every candidate in an analysis. The implementation records both the number and sizes of blocks, but no algorithm can infer independence from a CSV file. In particular, the public BRFSS and Chicago temporal tracks below are not retroactively certified by this proposition.
\end{remark}

\paragraph{Stress test of the audit unit.} We ran a fixed-seed 2,000-replicate diagnostic with 50 independent blocks of 200 perfectly dependent rows each. Within a block, all rows shared one Rademacher loss gap, so the true equally block-weighted gap was zero. Treating all 10,000 rows as iid gave empirical coverage of only 0.0975 (Monte Carlo standard error 0.0066) for the nominal 0.95 zero-gap target. Applying \cref{prop:blockcertificate} to the 50 independent block means gave coverage 0.9965 (Monte Carlo standard error 0.0013). Figure~S14 in the supplement is deliberately stylized: it validates the audit-unit distinction in a known dependence setting, not the dependence structure of any public data set.

\paragraph{Time-ordered audits.} Time-ordered audit records rarely offer independent blocks: consecutive rows are dependent and the dependence decays rather than vanishing at a known boundary. Under a declared dependence horizon $m_a$ (gaps more than $m_a$ positions apart are independent), blocks separated by buffers of at least $m_a$ rows are independent, and \cref{prop:blockcertificate} applies to the retained block means. The estimand changes (retained blocks are weighted equally and buffer rows are discarded), the horizon is declared and never inferred, and a mis-declared horizon invalidates the certificate. The statement, its proof, and a stress test with known dependence (row-level coverage $0.11$ against the nominal $0.95$; buffered-block coverage above $0.999$) are in Section~S24 of the supplement. The BRFSS and Chicago tracks were locked without a declared horizon and are not retroactively certified.

\subsection{How small can a tolerance be? The ideal-generator ceiling}
\label{sec:ceiling}

Locking a tolerance is necessary but not sufficient, because a tolerance chosen without reference to how much the declared loss varies between two samples of real data, and to how wide the audit interval is, can be either unattainable or vacuous. Section~S16 of the supplement calibrates the tolerance against a real-versus-real floor estimated on development data. That floor ignores the width of the interval, and the width is what an auditor actually pays. The following observation puts the two together.

\begin{proposition}[Ideal-generator ceiling]
\label{prop:ceiling}
Let the reference workflow be fitted on a real sample $R$ of size $m$, and let an \emph{ideal generator} return a sample $S$ of size $m$ that is independent of $R$ and has the same distribution. Let the audit sample be independent of both, and let the declared instrument at the declared look have outer end $U_n=|\widehat\Delta_n|+h_n$ (for the confidence sequence, $U_n=\max\{|\text{lower}|,|\text{upper}|\}$). Then $U_n$ for the ideal generator has the same law as $U_n^{\mathrm{rr}}$, the same quantity computed from two independent real samples, and the ideal generator is certified at tolerance $\tau$ with probability
\[
 \Pr(U_n\leq\tau)=\Pr(U_n^{\mathrm{rr}}\leq\tau).
\]
Consequently no tolerance below the $(1-\beta)$-quantile $\tau_\beta$ of $U_n^{\mathrm{rr}}$ certifies an ideal generator with probability above $1-\beta$.
\end{proposition}

\begin{proof}
The pair of fitted workflows $(f_R,f_S)$ is exchangeable with a pair fitted on two independent real samples of size $m$, and the audit gaps are computed on an independent audit sample, so the gap vector, and every functional of it, has the real-versus-real law.
\end{proof}

The proposition is elementary; its use is as a planning rule that needs no protected audit. The law of $U_n^{\mathrm{rr}}$ can be simulated from development data alone by fitting the workflow on pairs of disjoint real samples and scoring them on a pseudo-audit of development rows of the planned audit size. We did this on the Adult development pool (fitted size $1{,}000$, $400$ replicates, one candidate or a five-candidate panel, one look). \Cref{tab:ceiling} reports the result at the locked audit size of $16{,}000$ rows. Three findings follow. First, the range-based Hoeffding certificate cannot certify even an ideal generator at the recommended tolerance of $0.02$ (probability zero under both panels at $16{,}000$ pseudo-audit rows), and its $0.9$ ceiling is $0.025$ for one candidate and $0.029$ for five. Second, the empirical Bernstein and the anytime instruments certify an ideal generator at $0.02$ with probability one, with $0.9$ ceilings of $0.0068$ to $0.0076$ and $0.0059$ to $0.0064$. Third, the locked tolerance of $0.060$ was two to two and a half times the Hoeffding ceiling and nine to ten times the sequence's, which is consistent with the supplement's earlier statement that it was more permissive than a principled choice would have been. The same simulation at a pseudo-audit of $4{,}000$ rows (Section~S20) gives larger ceilings (Hoeffding $0.047$ to $0.055$, empirical Bernstein $0.014$ to $0.016$, sequence $0.010$ to $0.011$), which shows how quickly the ceiling falls with audit size. The pre-declared rule of the simulation (Section~S20) was that the recommended tolerance would be called unattainable if an ideal generator certified with probability below one half under the sequence at the pseudo-audit size; it did not. The $16{,}000$-row run was added after the $4{,}000$-row result had been seen, to match the locked audit size, and is labelled that way in the supplement.

\begin{table}[t]
\centering
\small
\caption{Ceiling for an ideal generator at the locked Adult audit size (single look, $16{,}000$ pseudo-audit rows from development data, $400$ replicates). ``$\tau_{0.9}$'' is the $0.9$ quantile of the outer end $U_n^{\mathrm{rr}}$; probabilities are the share of replicates with $U_n^{\mathrm{rr}}\leq\tau$.}
\label{tab:ceiling}
\begin{tabular}{llccccc}
\toprule
Panel & Instrument & Mean half-width & $\tau_{0.9}$ & $\tau=0.02$ & $\tau=0.03$ & $\tau=0.06$ \\
\midrule
1 candidate & Hoeffding & 0.0215 & 0.0249 & 0.00 & 1.00 & 1.00 \\
 & Empirical Bernstein & 0.0033 & 0.0068 & 1.00 & 1.00 & 1.00 \\
 & Anytime sequence & 0.0025 & 0.0059 & 1.00 & 1.00 & 1.00 \\
5 candidates & Hoeffding & 0.0257 & 0.0291 & 0.00 & 0.96 & 1.00 \\
 & Empirical Bernstein & 0.0041 & 0.0076 & 1.00 & 1.00 & 1.00 \\
 & Anytime sequence & 0.0031 & 0.0064 & 1.00 & 1.00 & 1.00 \\
\bottomrule
\end{tabular}
\end{table}

The ceiling is the level at which an ideal generator, one that outputs a fresh real sample, is certified. It is not a statement about any real generator, whose own bias adds to $U_n$, and a generator that copies its training rows can lie below it. In the Adult panel the selected copula's estimated gap is $0.014$, so its corrected sequence end of $0.0188$ is dominated by bias and not by interval width, which is why it certifies at $0.02$ while the ceiling for an ideal generator is $0.006$. Subsection~\ref{sec:fresh} uses the ceiling as the rule that sets the tolerance in a confirmatory study. The operative rule there is the $0.9$ quantile rounded up to the next $0.005$ (with a floor of $0.02$ in round 1): the rounding step, not the quantile, set every tolerance that bound, and it gives a candidate whose outer end lies between the quantile and the rounded value room to certify.

\section{Benchmark Design}
\label{sec:design}

\paragraph{Candidates and controls.} The completed studies use transparent baselines and a small set of learned models, not a contest for the best generator. The negative control independently resamples each feature and deliberately destroys dependence. The non-private comparator samples features conditionally on the target. The DP candidates use a target-conditional, independent-histogram mechanism with Laplace noise: under add/remove-one-record adjacency and fixed public domains, each mechanism releases one target histogram and ten target-by-feature contingency tables, and basic sequential composition gives pure $\eps$-DP with $\eps/11$ allocated to each release; clipping, normalization and sampling are post-processing. A regularized Gaussian copula is added in the California extension. Learned candidates are compact CTGAN and TVAE models \cite{xu2019ctgan}: a five-epoch budget in the Adult and Diabetes panels, and the library default of $300$ epochs in the confirmatory study. Every candidate in a case is fit on the same development subsample. The public sources are the ACS Public Use Microdata Sample \cite{census2024pums} and, from the UCI Machine Learning Repository, Bank Marketing \cite{moro2014bank}, Default of Credit Card Clients \cite{yeh2009default}, Online Shoppers Purchasing Intention \cite{sakar2018online}, Adult \cite{becker1996adult}, and Diabetes 130-US Hospitals \cite{clore2014diabetes}, with MAGIC Gamma Telescope added in the confirmatory study.

\paragraph{Studies.} \Cref{tab:studyinventory} gives the locked design of the studies. Four properties of the set matter more than any row. Every split is fixed before an audit row is read and no audit partition is used for fitting or selection, with one disclosed exception: the Adult audit had been opened once before its panel was locked (\cref{sec:adultpanel}), so the Adult certificate is nominal-only. Every panel contains the independent-marginal negative control, so each case can be checked for whether it discriminates at all (\cref{sec:controls}). Three records are retained as failures and not removed: the Bank Marketing and Diabetes cases are excluded because their negative controls passed the final grid, and an earlier exploratory CTGAN run is marked \code{INVALIDATED_NOT_RECORD_DISJOINT} because its claimed fresh replication overlapped the original sample by 28,098 rows. And the confirmatory studies of \cref{sec:fresh} and \cref{sec:roundtwo} are the only ones whose rules were fixed before any of their audit rows existed.

\begin{table}[t]
\centering
\footnotesize
\caption{Locked design of the main studies. Split sizes are development / selection / audit. ``Where'' gives the section reporting the result; S-numbers refer to the supplement. The Florida privacy battery, the BRFSS and Chicago tracks, the steel-energy challenge, a software sanity pilot on the scikit-learn breast-cancer table and the sealed prototype are in the supplement (S1, S9, S10, S21) and make no certificate claim.}
\label{tab:studyinventory}
\setlength{\tabcolsep}{3pt}
\begin{tabularx}{\linewidth}{p{0.20\linewidth}p{0.19\linewidth}Yp{0.13\linewidth}}
\toprule
Study & Splits & Candidate panel & Where \\
\midrule
ACS California frontier & 48,000 / 16,000 / 16,000 & Control, conditional marginal, DP histograms at $\eps\in\{0.5,1,4\}$; row-level, re-split by household post hoc in \cref{sec:fresh} & Main \\
ACS California extension & 48,000 / 16,000 / 16,000 & Above plus a regularized Gaussian copula; source-disjoint cohort & S4 \\
ACS New York grid & 50,000 / 20,000 / 30,000 & 33 candidates: three training sizes, three DP budgets, three seeds & S7 \\
ACS Texas & 60,000 / 20,000 / 20,000 & Control and conditional marginal, five checks; the control fails the Brier checks and the conditional marginal passes all five & Archive only \\
UCI Bank Marketing & 27,000 / 9,000 / 9,000 & Three candidates, five checks; excluded, control passed & S5 \\
UCI Default Credit & 15,000 / 6,000 / 9,000 & Three candidates, one primary check; sex field excluded & S5 \\
UCI Online Shoppers & 5,000 / 1,330 / 6,000 & Four candidates including a compact CTGAN & S6 \\
UCI Adult & 24,000 / 8,000 / 16,000 & Five-role panel with matched CTGAN and TVAE slots; audit opened once before the lock, nominal-only & Main \\
UCI Adult post-audit & Reuses the inspected audit & 300-epoch budget; three fit sizes at three seeds; nominal-only & S17 \\
UCI Diabetes 130 & 8,000 / 2,000 / 8,000 & Five-role panel; excluded, control passed & Main \\
Confirmatory: Pennsylvania, Illinois & about 38,000 / 12,700 / 12,700 persons, households as units & Ten candidates, 300-epoch CTGAN and TVAE with three seeds, ideal-generator positive control & \cref{sec:fresh}, S26 \\
Confirmatory: MAGIC & 9,510 / 3,804 / 5,706 rows & As above & \cref{sec:fresh}, S26 \\
Round 2: Ohio, Georgia & about 35,000 / 11,700 / 11,700 persons (Ohio), 33,500 / 11,200 / 11,200 (Georgia), households as units & Same ten candidates, two workflows (logistic regression, gradient boosting), no tolerance floor & \cref{sec:roundtwo}, S31 \\
\bottomrule
\end{tabularx}
\end{table}
\FloatBarrier

\section{Results}
\label{sec:results}

The confirmatory evidence comes first. Subsection~\ref{sec:controls} re-analyses the negative controls of every earlier locked case, which is what motivates the planning rules; Subsections~\ref{sec:fresh} and~\ref{sec:roundtwo} are the two confirmatory rounds; and the earlier locked studies on which the re-analysis draws (California, Adult, Diabetes) follow in Subsections~\ref{sec:california}, \ref{sec:adultpanel} and~\ref{sec:diabetescalibration}, kept in full because their records, including their failures, are part of the evidence. The ACS and UCI families are interleaved by purpose, not grouped by origin.

\paragraph{Reading the decisions.} Every reported pass is an equivalence certificate for the named loss, tolerance, population, and frozen workflow under the stated audit assumption. Every reported failure or hold means that this certificate was not obtained; it is not, by itself, a test establishing that the true gap exceeds tolerance. Negative controls are expected to be held, and a negative control that certifies equivalence invalidates the case's calibration role.

\subsection{When is a negative control informative?}
\label{sec:controls}

The protocol requires the dependence-destroying negative control to fail. In the locked records the control ``fails'' when its upper bound exceeds the tolerance, and a bound can exceed a tolerance simply because it is wide. We therefore re-analysed every locked case with a rule fixed in advance (Section~S22 and the revision protocol): the control is \emph{demonstrably outside tolerance} on the primary analysis if the \emph{lower} end of its interval exceeds the tolerance, and a case is \emph{informative} if that holds. The re-analysis recomputes each gap vector from the stored artifacts, reproduces every locked certificate to $10^{-9}$ before doing anything else, and applies three instruments to identical evidence: the locked Hoeffding half-width, the scheduled empirical Bernstein half-width, and the corrected anytime sequence. No locked decision is revised, and every recomputed number is a sensitivity result and not a certificate.

\begin{table}[t]
\centering
\small
\caption{Is the negative control demonstrably outside tolerance on the primary analysis? Entries are the lower end of the control's interval, with the case's tolerance in the second column; bold marks a lower end above the tolerance. H is the locked Hoeffding half-width, EB the scheduled empirical Bernstein half-width, CS the corrected anytime sequence. The primary analysis is predictive Brier loss for the ACS cases and the single Brier analysis elsewhere. The Texas CTGAN replication is the run later marked invalidated for overlap with the original sample; it is shown for completeness.}
\label{tab:controls}
\begin{tabular}{lcccc}
\toprule
Case & Tolerance & H & EB & CS \\
\midrule
California frontier & 0.050 & \textbf{0.0611} & \textbf{0.0854} & \textbf{0.0877} \\
California copula extension & 0.050 & \textbf{0.0579} & \textbf{0.0824} & \textbf{0.0846} \\
New York grid (training size 2,500) & 0.050 & \textbf{0.0705} & \textbf{0.0883} & \textbf{0.0873} \\
Texas (original audit) & 0.050 & \textbf{0.0602} & \textbf{0.0814} & \textbf{0.0826} \\
Texas (CTGAN replication) & 0.065 & 0.0364 & \textbf{0.0707} & \textbf{0.0777} \\
UCI Adult & 0.060 & 0.0326 & 0.0549 & 0.0573 \\
UCI Online Shoppers & 0.075 & 0.0000 & 0.0325 & 0.0401 \\
UCI Default Credit & 0.055 & 0.0000 & 0.0234 & 0.0267 \\
UCI Bank Marketing & 0.075 & 0.0000 & 0.0076 & 0.0126 \\
UCI Diabetes 130 & 0.080 & 0.0000 & 0.0000 & 0.0000 \\
\bottomrule
\end{tabular}
\end{table}

\Cref{tab:controls} shows a clean split. In every ACS case the control is demonstrably outside tolerance on the predictive Brier loss, under all three instruments in all but one case. In none of the five UCI cases is it. For Online Shoppers and Default Credit the control's point estimate ($0.049$ against a tolerance of $0.075$, and $0.035$ against $0.055$) is itself \emph{inside} the tolerance. The locked Hoeffding bound reported a failure only because the estimate plus its wide half-width ($0.051$ and $0.041$) exceeds the tolerance; the sharper intervals put the upper end below it ($0.066$ and $0.047$ under empirical Bernstein), so under the sharper instruments the control would have certified and those two cases would have been excluded, as Bank Marketing and Diabetes were. For Bank and Diabetes the estimates ($0.019$ and $0.003$) are far inside the tolerance: the workflow barely notices that dependence was destroyed. For Adult the control's estimate ($0.064$) exceeds the tolerance of $0.060$, but its lower end ($0.055$ to $0.057$) does not, so no instrument shows it outside tolerance at this audit size. A pass or hold on the locked Hoeffding bound can thus mix three different situations: a control that is clearly bad, a control that is bad but unproven at this audit size, and a control whose true gap is inside the tolerance because the tolerance is too loose for the workflow to separate it. The statement that ``the benchmark retains favorable, unfavorable and excluded outcomes'' is true, but the UCI outcomes labelled favorable were not shown to come from a case in which the control was shown to fail.

The ACS subgroup Brier checks are also demonstrably outside tolerance for the control under empirical Bernstein and the sequence, though not under Hoeffding (the exception is the Texas CTGAN replication, whose subgroup tolerance is $0.100$), and the ACS clipped-log-loss and income checks are not outside tolerance under any instrument; the full per-analysis tables are in Section~S22. We draw one conclusion for practice. A tolerance and an audit size should be planned together with the control: before the audit is opened, the custodian should be able to show, from development data, that the control would be demonstrably outside tolerance at the planned size (a development-only check is the subject of Section~S23), and a case in which that cannot be shown should be reported as uninformative.

\subsection{A pre-specified confirmatory study}
\label{sec:fresh}

The re-analysis above is post hoc, and the planning rules of \cref{sec:ceiling} and \cref{sec:controls} were derived from evidence that had been seen. We therefore ran one new study whose rules were written down and hashed before any of its audit rows existed. The protocol, one amendment, the case configurations, and the SHA-256 of every synthetic artifact and processed split are recorded in an append-only ledger in the repository before the audit partitions were read, and each audit was run once. The ledger is local and self-asserted; digests of the protocol files and of a ledger snapshot were later submitted to a public timestamping service, which anchors the record from that time on and does not itself establish the earlier order of events. ``Pre-specified'' in this paper therefore means hash-locked in that ledger, not registered with a third party before the audits. In both rounds the lock line and the stage-3 start line were written within eight seconds of each other and, for Illinois, the California re-split, Ohio and Georgia, in the same clock second, so the order rests on the sequence of ledger lines. Later cases were also locked after earlier audits had been read (MAGIC was audited before the Pennsylvania and Illinois locks, Ohio before Georgia's), with code and seeds unchanged.

\paragraph{Cases.} Three fresh cases, two chosen by a rule and one by a rule with a discretionary step: the two most populous states not yet used in the project, Pennsylvania and Illinois, from the 2024 ACS one-year PUMS person files, with \emph{households} (the public \code{SERIALNO}, the housing-unit identifier of the person file, so a person in group quarters is a one-person household) as the audit unit; and one public table that has never been run in the project, the UCI MAGIC Gamma Telescope table (binary target, numeric features), with rows as the audit unit; the rule for the table (more than $15{,}000$ rows, a binary target, numeric features, never run in the project) does not select a unique table, and MAGIC is the one we chose among those that qualify. Households are assigned to development, selection and audit by a seeded permutation, so no household straddles two partitions and persons follow their household. Both states had fewer eligible persons than the $80{,}000$ the planned partition sizes required ($63{,}945$ in Pennsylvania, $64{,}853$ in Illinois), so partitions follow a $60/20/20$ split of $99\%$ of the eligible persons; this fallback was recorded in an amendment, written from the counts alone, before any split was created. The realised audit partitions hold $7{,}677$ and $7{,}702$ households ($12{,}662$ and $12{,}841$ persons), and the MAGIC audit holds $5{,}706$ rows. For the ACS cases the estimand is the mean over audit households of each household's mean paired loss gap (\cref{prop:blockcertificate}); the person-level calculation, which ignores clustering, is also reported. We also re-split California by household from the raw archive used by the earlier studies, as a post hoc comparison with the original row-level audit.

\paragraph{Panel, instrument and tolerance.} Ten candidates per fresh case: the independent-marginal negative control; the target-conditional marginal; the Gaussian copula; CTGAN at the library default of $300$ epochs (ctgan 0.12.1 on torch 2.14, CPU, one thread) with three training seeds; TVAE at the same budget with three seeds; and an \emph{ideal-generator positive control}, a second disjoint real sample of the fitted size treated as if it were a synthetic table. All learned models are trained on a $4{,}000$-row development subsample, and the real reference workflow is fitted on the same subsample, so each certificate compares workflows trained on equal amounts of real data. The candidates and the reference are fitted on the same $4{,}000$ rows, whereas the ideal generator is drawn from other households, so the certificate measures agreement with the reference workflow trained on those rows and not fidelity to the population: a generator that reproduced its training rows would certify, and only the privacy layer, which lies outside the certificate, examines that. The primary instrument is the corrected anytime sequence at level $0.05/10$ on the audit-unit gap vector (primary analysis: overall Brier loss). The tolerance is set by the rule of \cref{sec:ceiling} from development data only: the $0.9$ quantile of the outer end over $200$ real-versus-real pseudo-audits, never below $0.02$, rounded up to the next multiple of $0.005$.

\paragraph{Decision rules, fixed before the audit was read.} The negative control must be \emph{demonstrably outside tolerance} (lower end above the tolerance), and the ideal-generator positive control must certify (outer end at most the tolerance); a case satisfying both is \emph{informative}, and candidates certify only in an informative case. Three outcome sentences were written in advance: if the control is demonstrably outside tolerance in fewer than two of the three fresh cases, the paper states that the tolerance-and-control rule did not confirm; if the positive control fails to certify in two or more, that the tolerance rule did not deliver an attainable tolerance; and if no non-control candidate certifies in any informative case, that no fresh candidate was certified at these audit sizes. A design-based sensitivity check (the $80$ ACS replicate weights) is reported for the ACS audits and decides nothing.

\paragraph{Results.} None of the three outcome sentences applies. Two of them carried little prior risk. With the tolerance at the floor of $0.02$ against a $0.9$ quantile of $0.005$ to $0.0175$, the ideal generator was almost certain to certify (all $200$ development replicates were below $0.02$ in every case). The negative control's lower end of $0.068$ to $0.091$ was far above the tolerance, which the development-only check of Section~S23 predicts before any audit is opened. The discriminating results are the candidates', below. The negative control was demonstrably outside tolerance in all three fresh cases, the ideal-generator positive control certified in all three, and all three cases were informative (\cref{tab:fresh}). The development-only rule produced the floor of $0.02$ in every case: the $0.9$ quantile of the outer end was $0.0050$ in both states and $0.0175$ in MAGIC, where only $1{,}510$ development rows remained for the pseudo-audit, so the MAGIC tolerance is looser than a size-matched rule would give. The pre-audit selection rule (smallest selection Brier among the non-control candidates) chose the Gaussian copula in all three cases, and in all three cases the copula certified.

\begin{table}[t]
\centering
\small
\caption{Confirmatory study, primary analysis (overall Brier loss), anytime sequence. ``Control lower'' is the lower end of the negative control's interval and ``Positive upper'' the outer end of the ideal-generator control's; ``Certified'' counts non-control candidates whose outer end is at most the tolerance. $\tau$ is the development-only tolerance.}
\label{tab:fresh}
\setlength{\tabcolsep}{2.5pt}
\begin{tabular}{lrcccccc}
\toprule
Case (unit) & Units & $\tau$ & Control lower & Outside? & Positive upper & Certifies? & Certified of 8 \\
\midrule
Pennsylvania (households) & 7,677 & 0.020 & 0.0912 & yes & 0.0039 & yes & 2 \\
Illinois (households) & 7,702 & 0.020 & 0.0776 & yes & 0.0036 & yes & 2 \\
MAGIC (rows) & 5,706 & 0.020 & 0.0676 & yes & 0.0045 & yes & 2 \\
\bottomrule
\end{tabular}
\end{table}

The candidates that certify are the transparent ones. In Pennsylvania and Illinois the Gaussian copula (outer ends $0.0064$ and $0.0065$) and the target-conditional marginal ($0.0122$ and $0.0131$) certify, and none of the twelve learned fits does: their outer ends range from $0.0203$ to $0.0372$, so the best one in Illinois misses the tolerance of $0.02$ by $0.0003$. In MAGIC the copula ($0.0066$) and one of three TVAE seeds ($0.0195$) certify; the other two TVAE seeds miss narrowly ($0.0242$ and $0.0212$), the target-conditional marginal misses ($0.0278$), and every CTGAN seed fails by a wide margin ($0.165$ to $0.232$). A learned model's outcome therefore depends on its seed, and the one learned pass is within $0.0005$ of the tolerance. Removing the $34$ MAGIC audit rows whose feature vectors also occur in development changes no certify decision (the TVAE pass moves from $0.0195$ to $0.0198$; Section~S26). This is a result at one training size (4{,}000 rows) and one compute budget, not a statement about either architecture.

Three sensitivity readings leave the conclusions unchanged. Under the scheduled empirical Bernstein instrument all three cases remain informative; under scheduled Hoeffding none does, because the ideal generator's outer end ($0.048$, $0.048$ and $0.055$) exceeds $0.02$, which is the main-article result of \cref{sec:ceiling} appearing in fresh data. The person-level calculation that ignores household clustering gives outer ends within $0.0013$ of the household calculation for every candidate that certifies, so clustering moved the intervals only slightly in these two states. And the design-based check is reassuring but descriptive. For every candidate in both states, $1.96$ standard errors from the $80$ replicate weights (between $0.0011$ and $0.0057$ for the primary analysis) is smaller than the half-width of the anytime sequence, so no ACS result is labelled not design-valid. That comparison sets a $95\%$ design interval against a $99.5\%$ sequence. At the sequence's own level the design half-width is within the sequence's for $59$ of the $60$ candidate-cell pairs of both rounds, the exception being a Pennsylvania CTGAN fit that does not certify, and the approximate design effect lies between $1.0$ and $2.8$ (Section~S31). The check says that the unit-level sequence is wider than a survey-weighted interval, not that it is a design-based inference.

\paragraph{What a point-estimate check would have decided.} A verification that reports the point estimate of the gap and accepts a candidate whose estimate is below the tolerance would have passed 13 of the $24$ non-control candidates across the three cases, including 7 of the $18$ learned fits, where the audit certified 6 and 1. The extra acceptances are the candidates whose estimate is below $0.02$ but whose simultaneous bound is not: in MAGIC the three TVAE seeds have estimated gaps of $0.0126$, $0.0168$ and $0.0139$, all below the tolerance, and the sequence certifies one of them; in Illinois four learned fits have estimates below $0.02$ and none is certified. This is the part of the decision that the panel correction, the audit schedule and the tolerance add to a point comparison: with ten candidates and several seeds, the best-looking estimates are optimistic, and the bound is what accounts for that.

\paragraph{California by household.} The post hoc re-split of California gives the same qualitative answer as the original row-level audit. The negative control's lower end is $0.0839$ against the original tolerance of $0.05$, and the target-conditional marginal and the three DP histograms certify with outer ends of $0.0161$ to $0.0177$ (the original row-level Hoeffding outer ends were $0.047$ to $0.048$, so the household sequence is much tighter; that comes from the instrument and not from the unit, which is now the household). The estimand is now a household mean, so the numbers are not directly comparable.

\paragraph{What the study does and does not show.} It shows, on audits that had not been opened and with a tolerance set before they were, that a tolerance of $0.02$ was attainable at these audit sizes (the ideal generator certified every time) and that the negative control could be shown outside it. It is not a severe test of the ceiling rule: the rule's data-driven term ($0.0050$, $0.0050$ and $0.0175$) never exceeded the $0.02$ floor, so the floor set every tolerance. As a post hoc sensitivity from the same results, without the floor the tolerance would have been $0.010$ in Pennsylvania and Illinois, where the ideal generator ($0.0039$, $0.0036$) and the copula ($0.0064$, $0.0065$) would still certify and the target-conditional marginal ($0.0122$, $0.0131$) would not; at the unrounded quantile ($0.0050$) the copula would not certify either. The $0.010$ rests on quantiles ($0.00502$ and $0.00503$) that lie $0.00002$ to $0.00003$ above the $0.005$ rounding step: bootstrapping the $200$ replicate outer ends gives a rounded tolerance of $0.005$, where the copula would not certify, in about $35\%$ to $44\%$ of resamples, so this counterfactual is close to a coin toss. The anytime intervals include audit-sampling variation only: they are conditional on one $4{,}000$-row development subsample and one split per case, on the frozen synthetic tables and on three training seeds. The learned models are at library defaults, untuned and at one training size, so this is not a matched-budget or tuned comparison. It does not show that the rule would work for a different workflow, a different loss, or a tolerance imposed from outside; the three cases share one logistic-regression Brier workflow, and two of them are the same survey in two states. \Cref{sec:roundtwo} removes the floor and adds a second workflow. The two generators that certify both condition on the target that the workflow predicts, which is alignment with the workflow and not evidence of general fidelity. The ACS partitions are smaller than planned, the MAGIC tolerance is looser than a size-matched rule would give, and the learned models are a single compute budget at one training size. Two post hoc comparisons on development data, neither used to choose any analysis above, are in the supplement: the standard utility measures (a propensity-score ratio and confidence-interval overlap) rank the candidates in the same direction as the certificate but give no decision, and disagree with its certify calls on two of six case-and-measure pairs (Section~S29); and an equivalent real-sample size was small and mostly zero, so it is not offered as a rule (Section~S28).

\subsection{A second pre-specified round: no floor, a second workflow}
\label{sec:roundtwo}

The first round was not a severe test of the tolerance rule, because the rule's data-driven term never exceeded the floor. We therefore ran a second round, whose protocol (Part D) was written and hashed in the same self-attested, locally logged sense before any of its files existed; an amendment fixing readings the protocol left open was recorded before any split was created. It removes the floor, adds a second workflow, puts six predictions from round 1 on record and states four outcome sentences in advance.

\paragraph{Design.} The cases are Ohio and Georgia, the next two most populous states not yet used, from the same 2024 ACS one-year person files, with households as audit units. Both again had fewer eligible persons than the planned $80{,}000$ ($58{,}963$ and $56{,}407$), so the pre-declared $60/20/20$ fallback applied; the audit partitions hold $7{,}213$ and $6{,}699$ households ($11{,}675$ and $11{,}168$ persons). The panel is that of round 1, and the same ten synthetic tables are scored by two workflows: W1, the logistic-regression Brier workflow of round 1, and W2, gradient boosting with the same features (scikit-learn's histogram booster, $100$ iterations, trained with the library's log-loss objective and evaluated with the Brier loss). A state-and-workflow pair is a \emph{cell} with its own tolerance, controls and decisions, giving four cells, and the anytime sequence runs at level $0.05/20$. The tolerance is the $0.9$ quantile of the outer end over $200$ development-only real-versus-real pseudo-audits, rounded up to a multiple of $0.005$, with a lower limit of $0.005$ and no $0.02$ floor. The outcome sentences were: the ideal-generator control fails to certify in two or more cells (K4, ``the rule did not deliver an attainable tolerance''); the negative control is not demonstrably outside tolerance in two or more cells (K5); no non-control candidate certifies in any informative cell (K6); the copula fails to certify in every informative cell with $\tau\geq0.010$ (K7).

\paragraph{Results.} None of the four outcome sentences applies. The rule now bound: the $0.9$ quantiles were $0.0054$ (Ohio) and $0.0062$ (Georgia) for W1 and $0.0106$ and $0.0107$ for W2, so the tolerances were $0.010$ and $0.015$, set by the data quantile plus a rounding margin of $0.0038$ to $0.0046$ and not by a floor (\cref{tab:roundtwo}). All four cells were informative: the negative control's lower end was between $0.077$ and $0.090$, and the ideal generator certified in every cell (outer ends $0.0042$ to $0.0091$). The Gaussian copula certified in all four cells (outer ends $0.0069$, $0.0115$, $0.0079$ and $0.0128$). The target-conditional marginal certified in one cell, Georgia with W1 ($0.0094$ against $0.010$), and missed in the other three ($0.0118$, $0.0239$ and $0.0179$). None of the $12$ learned tables certified under either workflow (the same table is scored by W1 and W2, so the $24$ cell results are not independent); their outer ends run from $0.0199$ (CTGAN seed 1, Ohio, W1, about twice the tolerance) to $0.0534$.

\begin{table}[t]
\centering
\small
\caption{Second round, primary analysis (overall Brier loss), anytime sequence at level $0.05/20$. $q_{0.9}$ is the $0.9$ quantile of the outer end over $200$ development-only real-versus-real pseudo-audits; the tolerance $\tau$ is $q_{0.9}$ rounded up to a multiple of $0.005$, with no floor. ``Control lower'' and ``Positive upper'' are as in \cref{tab:fresh}. W1 is the logistic-regression workflow of the first round and W2 gradient boosting. All four cells are informative.}
\label{tab:roundtwo}
\setlength{\tabcolsep}{4pt}
\begin{tabular}{lrccccp{0.25\linewidth}}
\toprule
Cell & Households & $q_{0.9}$ & $\tau$ & Control lower & Positive upper & Certified candidates \\
\midrule
Ohio, W1 & 7,213 & 0.0054 & 0.010 & 0.0859 & 0.0042 & copula \\
Ohio, W2 & 7,213 & 0.0106 & 0.015 & 0.0864 & 0.0078 & copula \\
Georgia, W1 & 6,699 & 0.0062 & 0.010 & 0.0774 & 0.0056 & target-conditional marginal, copula \\
Georgia, W2 & 6,699 & 0.0107 & 0.015 & 0.0899 & 0.0091 & copula \\
\bottomrule
\end{tabular}
\end{table}

\paragraph{What was and was not at risk.} The copula's outer end exceeded the $0.9$ quantile in all four cells ($0.0069$ against $0.0054$, $0.0115$ against $0.0106$, $0.0079$ against $0.0062$ and $0.0128$ against $0.0107$), as did the target-conditional marginal's in Georgia with W1 ($0.0094$ against $0.0062$). In both rounds no candidate's outer end in an ACS cell lay below that cell's quantile. The copula and the one marginal therefore certified only because the tolerance is rounded up to the next $0.005$: at the unrounded quantile no candidate would have certified in any cell of round 2, and K6 would have been met. We kept the rounding rule from round 1 knowing this (Part D records the round-1 copula outer ends, $0.0064$ and $0.0065$, and quantiles on the edge of the $0.005$ step). The ideal-generator and negative-control outcomes were close to certain in advance (all $200$ development replicates were below the tolerance in every cell, and the control's lower end of $0.077$ to $0.090$ was far above tolerances of $0.015$ or less), so K4 and K5 could hardly fire. The outcomes that carried risk were the copula's and the marginal's margins (Pr4 and Pr5) and the learned tables.

Of the six predictions on record, five were borne out (the tolerances were $0.010$ or smaller for W1; the control was outside tolerance and the ideal generator certified in every cell; the copula certified wherever $\tau\geq0.010$, which was every cell; no learned model certified). One was not: Pr5 said the target-conditional marginal would not certify at $\tau\leq0.010$, and it did in Georgia, by $0.0006$. We report that as a failed prediction. The round-1 outer ends ($0.0122$ and $0.0131$) were close enough to $0.010$ that either outcome was possible.

Three sensitivity readings (Section~S31) leave the conclusions in place. The scheduled empirical Bernstein bound leaves three of four cells informative and certifies no candidate in any, so under it the study would have said less, not something different; scheduled Hoeffding leaves no cell informative, because its ideal-generator outer end is $0.051$ to $0.055$, which is the main-article result of \cref{sec:ceiling} appearing again. The person-level sequence, which ignores household clustering, certifies the same candidates, with outer ends within $0.0012$ of the household calculation, and no candidate's survey-weighted interval exceeds the sequence's half-width (a $95\%$ design interval against a $99.5\%$ or $99.75\%$ sequence; the matched-level comparison and an approximate design effect are in Section~S31, and the check is descriptive). A point-estimate check would have accepted $7$ of the $32$ non-control candidates against $5$ certified, and none of the $12$ learned tables under either workflow; the two extra acceptances are the target-conditional marginal in Ohio with W1 and in Georgia with W2. Unlike round 1, the bound therefore changed a decision only for a transparent candidate: with the floor removed, no learned fit had an estimated gap below the tolerance (the smallest was $0.0134$).

\paragraph{What the second round does and does not show.} It shows that the unfloored rule delivered attainable tolerances ($0.010$ and $0.015$) for two workflows in two further states, that the negative-control rule held in all four cells, and that the copula's round-1 certification did not depend on the $0.02$ floor, although it does depend on the $0.005$ rounding step. It is still a limited test. The ideal generator certifying in every cell is close to guaranteed by construction, since the tolerance is rounded up from a $0.9$ quantile to between $1.4$ and $1.9$ times that quantile, so it is not strong evidence on its own. Given how the tolerance is built, the check only tests that the development pseudo-audits and the audit are exchangeable and that the simulation code matches the audit code. It is not evidence that the rule is good. The ideal sample is drawn from households the reference did not use, which over-represents small households, whereas \cref{prop:ceiling} assumes the same distribution. The tolerance simulation uses the same scheme, so the two agree with each other, and the ideal generator's estimated gaps were at most $0.002$ in absolute value in round 2, so the effect looks negligible. The two states are the same survey in the same year as round 1, the panel, code and learned-model budget are the same, and the authors wrote the rule, the workflows and the predictions. The tolerance is still one the authors' rule chose, not one an agency set, and the rule's step of $0.005$ is the authors' choice. The round-2 quantiles lie far from a step edge (bootstrap standard errors $0.0001$ to $0.0002$) but the round-1 quantiles did not, and an agency would need a stated reason for the step. The four cells are two states scored by two related workflows on the same tables, so they are two replications and not four, and the two generators that certified both condition on the target that the workflows predict. Part D was written with the round-1 outer ends in hand, and Pr1 and Pr3 were wide forecasts. The outcome sentences K4 and K5 are consistency checks, and the decisions are taken at the final look only, so the time-uniform property of the sequence was not used to stop early. The learned fits are library defaults at one training size, and the intervals condition on one development subsample per cell and three training seeds. In the post-audit Adult scaling arm (Section~S17; nominal-only bounds) the same architectures pass at $16{,}000$ development records, so these failures concern $4{,}000$ training rows and library defaults.

\FloatBarrier

\subsection{California: a locked audit with a flat frontier}
\label{sec:california}

The selection Brier losses were 0.2502 for the independent-marginal control, 0.1703 for the conditional non-private comparator, and 0.1701, 0.1708, and 0.1702 for DP budgets 0.5, 1, and 4, respectively; the real-development reference was 0.1562. At the final 16,000-row scheduled audit look, the independent-marginal negative control failed the five-analysis grid. The conditional non-private comparator and all three DP candidates passed. Table~\ref{tab:ca} gives the complete final-look summary, Figure~\ref{fig:caoverview} separates selection from audit outcomes; the fixed-budget selection comparison across the three DP budgets appears in Section~S18 of the supplement. This is the intended behavior of the audit: it rejects a dependence-destroying baseline and reports a bounded comparison for a declared use. It is not evidence that the passing candidates are safe to release. Two limits of the study should be stated. The four passing candidates have almost identical estimated Brier gaps ($0.0126$ to $0.0131$), so the study shows no frontier among them: the non-private conditional comparator and the three DP histograms are statistically indistinguishable at this audit size. And the descriptive-income check has an estimated gap of $0.0000$ for every candidate, including the negative control, so it cannot discriminate and contributes nothing to the decision.

\begin{table}[t]
\centering
\small
\caption{California final-look summary. ``Worst upper gap'' is the largest $|\widehat{\Delta}|+h$ among the five declared checks; direct comparisons across analyses should still use each analysis's own tolerance.}
\label{tab:ca}
\begin{tabular}{lccc}
\toprule
Candidate & Selection Brier & All five checks pass? & Worst upper gap \\
\midrule
Independent marginal, non-private & 0.2502 & No & 0.1471 \\
Target conditional, non-private & 0.1703 & Yes & 0.0621 \\
DP conditional histogram, $\eps=0.5$ & 0.1701 & Yes & 0.0622 \\
DP conditional histogram, $\eps=1$ & 0.1708 & Yes & 0.0624 \\
DP conditional histogram, $\eps=4$ & 0.1702 & Yes & 0.0621 \\
\bottomrule
\end{tabular}
\end{table}

\begin{figure}[t]
\centering
\includegraphics[width=0.97\linewidth]{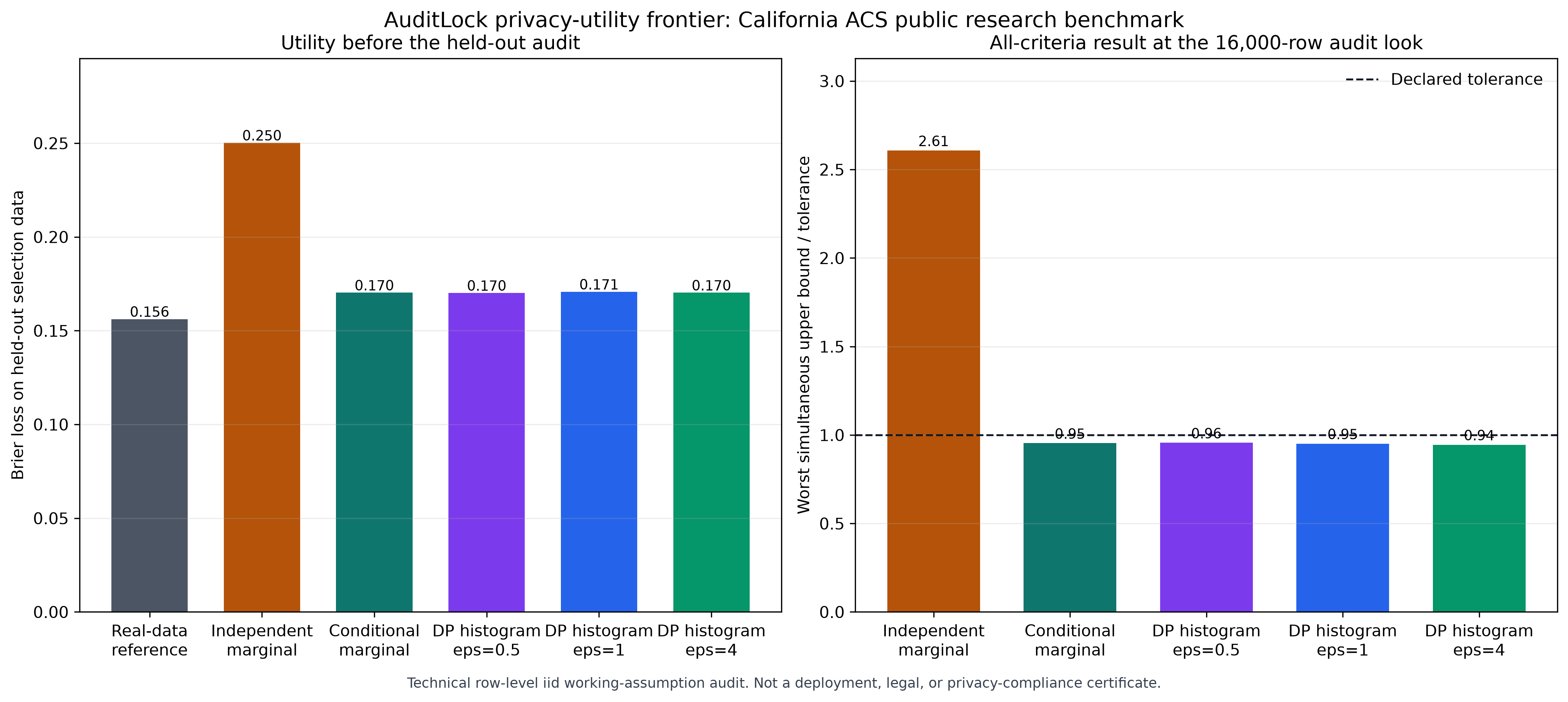}
\caption{California locked-audit overview. The negative control fails while the target-conditional and DP candidates pass the declared public technical checks at the final scheduled look. This figure is a use-specific diagnostic under the study assumptions, not a release authorization.}
\label{fig:caoverview}
\end{figure}

\subsection{UCI Adult: a five-role panel on an audit opened once before the lock}
\label{sec:adultpanel}

\paragraph{Disclosure.} Earlier on the day this panel was locked, a four-candidate, 16,000-record audit of the same protected Adult audit file, at the same tolerance of $0.060$, was computed by a run that was later invalidated, because a stale process and a newer process wrote to the same result directory. Its final-look upper gaps were $0.109$ for the independent-marginal control, $0.054$ for the target-conditional marginal, $0.054$ for the Gaussian copula and $0.066$ for a CTGAN fit in the stale process (a longer, 50-epoch budget). That certificate, and a figure drawn from it, existed before the panel reported below was locked. What was fixed earlier, in a configuration written before that inspection (the five-epoch configuration is time-stamped 19:03 and the certificate 19:06 local time on 19 July, so a short interim view cannot be excluded from the time stamps alone): the tolerance, the single Brier analysis, the three scheduled looks, and the compact five-epoch, 1{,}000-record fit budget. What changed afterwards: a matched TVAE was added, and the training-subsample and audit-order seeds were replaced. By the paper's own rule for a previously inspected audit, the Adult certificate is therefore nominal-only. We keep the case because it is the only study in the project that runs the complete five-role panel, but nothing in the paper's argument for a certificate rests on it; the California studies and the confirmatory study of \cref{sec:fresh} do that work, and the Adult numbers should be read as a measurement.

The panel runs the challenge protocol's complete candidate set, with CTGAN and TVAE trained under the identical CPU budget. On the 8,000-record selection partition, the Brier losses were 0.1750 for the independent marginal, 0.1254 for the target-conditional marginal, 0.1206 for the Gaussian copula, 0.2323 for the small CTGAN, and 0.2333 for the small TVAE, compared with 0.1084 for the real-development reference. The Gaussian copula had the best selection score, so it was chosen before the final audit.

At the final 16,000-record audit look, the simultaneous Brier upper gaps were 0.0959, 0.0482, 0.0457, 0.1556, and 0.1518 for the same five candidates, against the locked tolerance of 0.060 (Figure~\ref{fig:adultpanel}). The negative control failed, which the certificate requires. The target-conditional marginal and the selected Gaussian copula passed. Both small learned models failed, with estimated gaps of $+0.124$ for CTGAN and $+0.120$ for TVAE. The two learned models land close to each other and clearly apart from the simple baselines, which is exactly the comparison the matched-slot rule is meant to show. On CPU, CTGAN took 58.5 seconds to fit and sample; TVAE took 4.4 seconds. As with the Online Shoppers test, this is negative evidence under one small declared budget. It does not say that either model fails with more training, other data, or other implementations. The case is control-sensitive and its release packet stays blocked until a person makes the required integrity attestations.

\begin{figure}[t]
\centering
\includegraphics[width=0.98\linewidth]{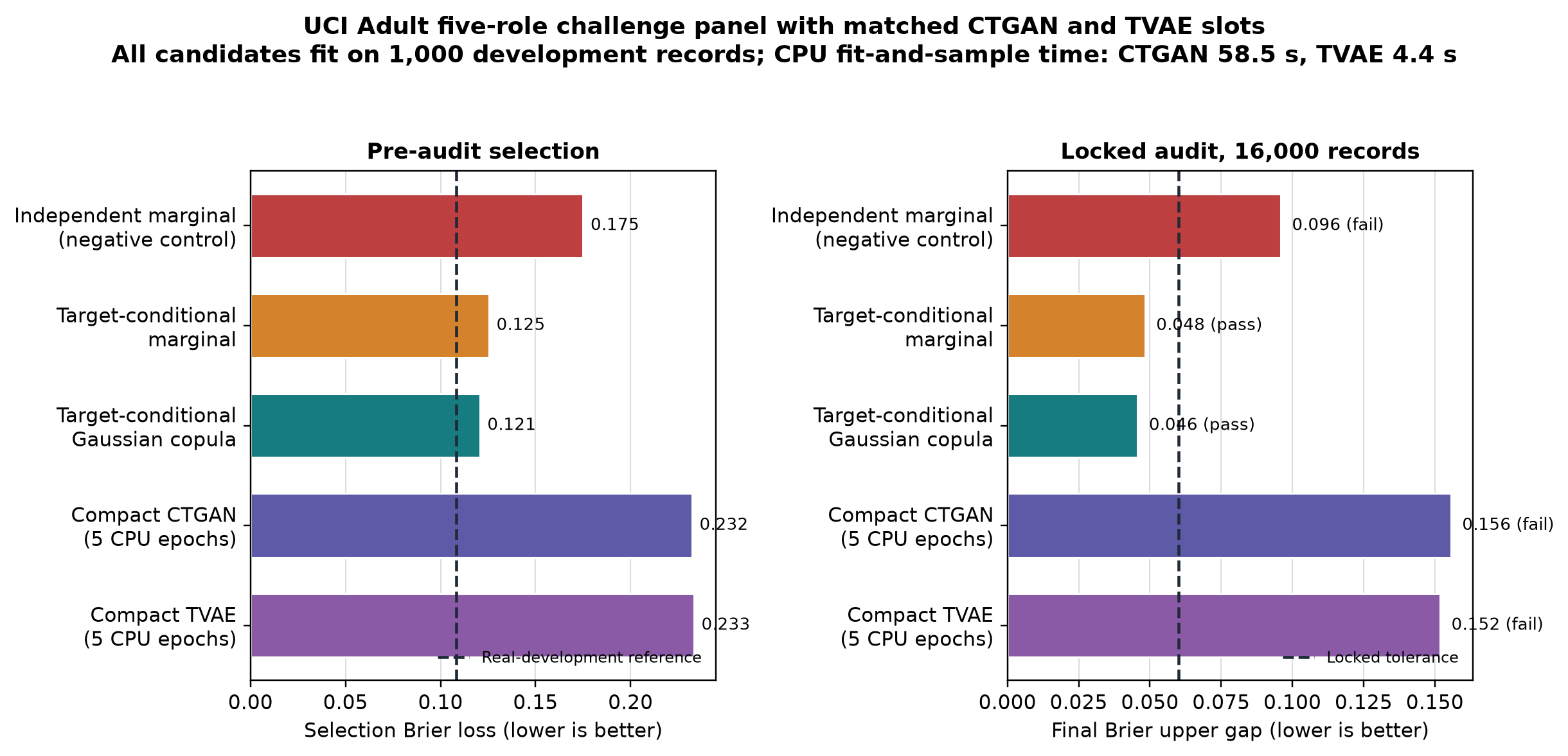}
\caption{UCI Adult five-role challenge panel with matched CTGAN and TVAE slots. Every candidate was fit on the same 1,000-record development subsample under the shared five-epoch CPU budget. Left: the selection comparison before the audit. Right: the final simultaneous Brier upper bound at 16,000 audit records. The negative control and both small learned models fail the locked check; the simple conditional baselines pass. This is a small-budget research measurement, not a model ranking or permission to release data.}
\label{fig:adultpanel}
\end{figure}
\FloatBarrier

\paragraph{Post-audit arms.} Three further arms reuse the already-inspected
Adult audit and therefore carry nominal-only bounds and make no certificate
claim: a 300-epoch budget-sensitivity arm, a data-scaling arm refitting the panel
at 1,000, 4,000, and 16,000 development records, and a three-seed repetition of
that sweep. They are reported in Section~S17 of the supplement. The headline is
that the same locked criterion admits both learned models once they are fit on
enough data, while still rejecting the dependence-destroying control at every
size, so the criterion discriminates rather than merely rejects. The three-seed
repetition also withdraws a finer single-seed reading: CTGAN overtakes the copula
at 16,000 records in only two of three seeds, so that ordering is not claimed.

\subsection{UCI Diabetes 130-US Hospitals: an excluded health calibration}
\label{sec:diabetescalibration}

The Diabetes 130-US Hospitals study deliberately records a calibration failure rather than trying to rescue it after the audit. On the 2,000-record selection partition, Brier losses were 0.0788 for the independent marginal, 0.0795 for the target-conditional marginal, 0.0764 for the Gaussian copula, 0.2478 for compact CTGAN, and 0.2127 for compact TVAE, compared with 0.0746 for the real-development reference. The Gaussian copula was selected before audit. At the final 8,000-record look, simultaneous Brier upper gaps were 0.0475, 0.0517, 0.0453, 0.2093, and 0.1782 for the same candidates, against the locked tolerance of 0.080 (Section~S18 of the supplement). CTGAN and TVAE both fail under this matched compact budget. For the benchmark, what matters more is that the independent-marginal negative control also passes. The threshold therefore does not discriminate the declared dependence-destroying control in this case. We retain the complete locked record as \code{EXCLUDED\_NEGATIVE\_CONTROL\_PASSED}, do not retune the tolerance, and do not count the case as favorable health-domain or clinical evidence.

\FloatBarrier

\section{What the Benchmark Does Not Show}
\label{sec:limits}

The paper's claims are deliberately narrow.

\paragraph{Evidence.} The row-level studies cover one public survey family and five UCI tables, and the Bank Marketing and Diabetes 130-US Hospitals cases remain excluded calibration failures. Every demonstration is a unit-level technical diagnostic, row-level except for the confirmatory rounds, which split and audit by household: none validates design-based survey inference, temporal transfer, a confidential release setting, or any credit, advertising, clinical or consumer-targeting decision. The generators are simple marginals, a regularized Gaussian copula, exact-DP histograms and fixed-budget learned models, which is not a claim to cover modern synthesis, and learned candidates are single locked artifacts or three seeds, so they are evidence about those artifacts and not about how a generator behaves on average. The post-audit Adult arms reuse an inspected audit and carry nominal-only bounds (Section~S17). The privacy layer is a small named attack battery whose scope and limits are in Sections~S8 and~S21; a low score under it is not proof of non-identifiability.

\paragraph{The certificate itself.} It requires a locked panel and a valid audit split, and it does not repair model selection or p-hacking performed after audit access. It measures agreement with the reference workflow trained on the same development rows as the candidates, not fidelity to the population. The independent-block certificate applies only where the sampling design justifies independent blocks, and the buffered-block certificate of Section~S24 inherits a declared dependence horizon that no algorithm can verify from data. The variance-adaptive and anytime instruments of \cref{sec:varianceadaptive} are a measurement on already-locked evidence: they revise no recorded outcome, and choosing a bound family after seeing which one passes is exactly the adaptation \cref{prop:reuseimpossibility} forbids.

The repository retains failures and invalidations: the New York low-scale failures, the invalidated exploratory CTGAN run, and the sealed demonstration that is blocked despite its passing technical grid. An evaluation record should make it harder, not easier, to hide unfavorable or ineligible evidence.

\section{Reproducibility and Adoption}

The public package contains locked configurations, source provenance, candidate metadata, synthetic artifacts, checksums, figures, implementation code, unit tests and run guides for every completed study. It is archived under CC BY 4.0 at Zenodo \cite{opoku2026synthguardsoftware}. Version 0.1.15 accompanies this manuscript; the concept DOI \href{https://doi.org/10.5281/zenodo.21503396}{10.5281/zenodo.21503396} always resolves to the latest version, and anyone reproducing these results should use the version DOI, because every number reported here is checked against the locked records in that exact version. The package installs an \code{auditlock} command with audit planning, bounded-gap auditing and fail-closed registry validation, and ships a container, a Python 3.11/3.13 continuous-integration workflow, a generator-adapter contract and a machine-readable release registry. Its tests cover split integrity, the scheduled row and block bounds, the variance-adaptive and anytime-valid bounds (including, since the correction, a coverage test on two-point data at the range endpoints), the reconstruction attack, the certification-margin planner, the audit-reuse boundary, DP histogram composition, attack calibration, release-gate behavior, and the confirmatory pipeline (which does not read an audit partition before the lock: a path guard in the project's own read functions, and a test that stage 2 gives an identical lock when the audit file is replaced by garbage).

\paragraph{Reruns.} Two non-author reruns of the immutable v0.1.11 archive, concerning the Online Shoppers case, were received, from two distinct institutions on distinct machines; both reproduced the substantive aggregate result with small disclosed deviations and each record has a documented completeness gap. The strict two-record milestone is therefore not claimed, and since only one institution is external to the authors, independent reproduction across two external institutions is not claimed either. Names, signed forms and machine reports are held in a private evidence archive available to editors on request. No external rerun of either confirmatory round has been received. The archive contains a rerun kit with pass criteria written beforehand; the authors' own run of it, on the original machine, reproduced the results exactly, which shows determinism on one machine and says nothing about others.

For an institution, adoption should follow a fixed sequence: define the permitted
use and tolerances; establish training, selection, and audit separation at the
appropriate person, household, site, or time unit; lock the candidate and
analysis panels; run the audit inside the organization's controlled environment;
review the complete evidence packet; and only then consider a release decision
under the organization's own legal, privacy, and governance requirements.
\sg{} supplies technical evidence; it does not replace the organization
responsible for the release.

\section*{Data Availability Statement}

The public software, benchmark configurations, aggregate result records, tests, container files, and reproduction instructions are archived as \sg{} version 0.1.15, which contains the corrected confidence sequence, at \href{https://doi.org/10.5281/zenodo.23736369}{10.5281/zenodo.23736369}. The earlier version 0.1.14 (\href{https://doi.org/10.5281/zenodo.21909680}{10.5281/zenodo.21909680}) contains the uncorrected sequence and should not be used for the anytime results. The public archive does not contain reviewer forms, signatures, credentials, confidential audit rows, or restricted row-level data; it contains the authors' contact details in the unanonymised sources and, in a few result records, local file paths. Public source data are obtained from the cited providers using the documented retrieval and integrity-check procedures.

\section{Conclusion}

Synthetic-data evaluation needs a defensible middle ground between ``looks realistic'' and ``safe to release.'' \sg{} provides one: declare the intended use and tolerance, lock the candidate and analysis panels, evaluate them on a protected audit with simultaneous uncertainty, calibrate the workflow with controls, report privacy evidence for a named threat model, and preserve a separate human release decision.

The empirical record is more modest than a count of passes suggests. Some transparent candidates certify for declared uses and compact learned models do not under the tested budgets, but these outcomes are use-specific and are not a generator ranking. What the audit can say depends at least as much on the tolerance and on how the negative control is used as on the concentration inequality: with the locked tolerances only the ACS predictive-Brier controls were shown to lie outside tolerance, and in the five UCI cases the control either sat inside the tolerance or could not be separated from it at the audit size. A tolerance and an audit size should therefore be planned together, from development data, before the audit is opened: the ideal-generator ceiling says how small a tolerance can be at a given size, and a development-only check says whether the control can be shown outside it.

The confirmatory rounds, pre-specified and hash-locked before their audits were opened and publicly timestamped only afterwards, show that tolerances fixed from development data were attainable and that the control could be shown outside them; both outcomes were close to certain by design. What discriminated is narrower. A Gaussian copula certified in every ACS cell, but only inside the rounding step of the tolerance rule (at the unrounded quantile no ACS candidate would have certified). A target-conditional marginal certified once, against a written prediction. Learned models, at library defaults, one training size and one compute budget, certified in one of eighteen fits in round 1 and in none of twelve tables in round 2. The studies do not show that the rule transfers to other data families, losses, or a tolerance set by a data custodian, and they do not rank generators.

The contribution is a protocol and an evidence record built from standard inequalities, not a new probability theorem. It is a pre-specified, simultaneous form of the verification servers of official statistics, a planning rule for tolerances and controls, a correction to an instrument we had mis-implemented, and a record that keeps its failures: a corrected confidence sequence, an audit that had been opened once before its lock, a sweep that stopped short, and cases in which the control could not be separated from the reference. A locked audit earns trust not by making candidates pass, but by making unsupported release claims difficult to produce.

\section*{Acknowledgements}

We thank the two colleagues who independently reran the released benchmark on
their own machines and returned their results, including the deviations they
observed. Neither is an author of this paper, neither had access to the protected
audit rows, and neither participated in any release decision. Their identities,
signed forms, and machine reports are held in a private evidence archive and are
available to the editors on request.

\section*{Funding}

The authors received no funding for this work.

\section*{Declaration of Conflicting Interests}

The authors declare that there is no conflict of interest.

\end{document}